\documentclass[11pt]{article}
\usepackage[abs]{overpic}
\usepackage[utf8]{inputenc}
\usepackage{a4wide}
\usepackage{graphicx}
\usepackage[hmarginratio={1:1},     
  vmarginratio={1:1},     
  textwidth=17cm,        
  textheight=24cm,
  heightrounded,]{geometry}
  
\usepackage{amsmath,accents}
\usepackage{amsthm}
\usepackage{amssymb}

\usepackage{mathrsfs, mathtools}

\let\emph\undefined
\newcommand{\emph}[1]{\textsl{#1}}

\usepackage[hidelinks]{hyperref}

\numberwithin{equation}{section}
\usepackage{mathtools}
\mathtoolsset{showonlyrefs}
\newtheoremstyle{style1}
  {13pt}
  {13pt}
  {}
  {}
  {\normalfont\bfseries}
  {.}
  {.5em}
  {}
\theoremstyle{style1}

\newtheorem{definition}[equation]{Definition}

\newtheorem{remark}[equation]{Remark}

\usepackage{tocloft}
\usepackage{enumitem}

\newtheoremstyle{style2}
  {13pt}
  {13pt}
  {\slshape}
  {}
  {\normalfont\bfseries}
  {.}
  {.5em}
  {}

\theoremstyle{style2}

\newtheorem{lemma}[equation]{Lemma}
\newtheorem{theorem}[equation]{Theorem}
\newtheorem{proposition}[equation]{Proposition}

\usepackage{dsfont}
\usepackage{tikz}
\usetikzlibrary{matrix,arrows,decorations.pathmorphing}
\usepackage{tikz-cd}

\usepackage{tikz}

\newcommand{\R}{\mathbb{R}}
\newcommand{\C}{\mathbb{C}}
\newcommand{\Z}{\mathbb{Z}}

\newcommand{\N}{\mathbb{N}}

\newcommand{\Ca}{\mathcal{C}}

\newcommand{\Ra}{\mathcal{R}}

\newcommand{\Aa}{\mathcal{A}}

\newcommand{\Ia}{\mathcal{I}}
\newcommand{\Ma}{\mathcal{M}}
\newcommand{\Pa}{\mathcal{P}}
\newcommand{\Oa}{\mathcal{O}}

\newcommand{\Hscr}{\mathscr{H}}

\newcommand{\Sscr}{\mathscr{S}}
\newcommand{\Wscr}{\mathscr{W}}

\newcommand{\rmB}{\mathrm{B}}

\newcommand{\rmt}{\mathrm{t}}

\newcommand{\dd}{\mathrm{d}}

\newcommand{\bfB}{\vec{B}}
\newcommand{\bfdd}{\boldsymbol{\mathrm{d}}}

\newcommand{\sft}{{\mathsf{t}}}

\newcommand{\sfG}{{\mathsf{G}}}
\newcommand{\sfU}{{\mathsf{U}}}
\newcommand{\sfC}{{\mathsf{C}}}
\newcommand{\sfA}{{\mathsf{A}}}
\newcommand{\sfM}{{\mathsf{M}}}

\newcommand{\hol}{\mathrm{hol}}

\newcommand{\End}{\mathsf{End}}
\newcommand{\Hom}{\mathsf{Hom}}

\newcommand{\id}{\text{id}}
\newcommand{\Pic}{\text{Pic}}
\newcommand{\HVBdl}{\mathrm{HVBdl}^\nabla}
\newcommand{\HLBdl}{\mathrm{HLBdl}}
\newcommand{\BGrb}{\mathrm{BGrb}}

\newcommand{\iu}{\mathrm{i}}

\newcommand{\frh}{\mathfrak{h}}
\newcommand{\frM}{{\mathfrak{M}}}

\newcommand{\triv}{\mathrm{triv}}
\newcommand{\Mat}{\mathrm{Mat}}
\newcommand{\One}{\mathds{1}}
\newcommand{\qandq}{\qquad \text{and} \qquad}
\newcommand{\rk}{\mathrm{rk}}

\newcommand{\Hilb}{{\mathsf{Hilb}_\C}}

\newcommand{\e}{\,\mathrm{e}\,}
\newcommand{\ii}{\,\mathrm{i}\,}

\newcommand{\mbf}{\boldsymbol}
\let\to\undefined
\newcommand{\to}{\longrightarrow}
\let\mapsto\undefined
\newcommand{\mapsto}{\longmapsto}

\newcommand{\opp}{\text{op}}
\newcommand{\<}{\langle}
\renewcommand{\>}{\rangle}
\newcommand{\rmpar}{\mathrm{par}}

\DeclareMathSymbol{\Phiit}{\mathalpha}{letters}{"08} 
\DeclareMathSymbol{\Psiit}{\mathalpha}{letters}{"09}
\DeclareMathSymbol{\Sigmait}{\mathalpha}{letters}{"06}
\DeclareMathSymbol{\Xiit}{\mathalpha}{letters}{"04}
\DeclareMathSymbol{\Piit}{\mathalpha}{letters}{"05}\let\Pi\undefined\newcommand{\Pi}{\Piit}
\DeclareMathSymbol{\Gammait}{\mathalpha}{letters}{"00}
\DeclareMathSymbol{\Omegait}{\mathalpha}{letters}{"0A}
\DeclareMathSymbol{\Upsilonit}{\mathalpha}{letters}{"07}
\DeclareMathSymbol{\Thetait}{\mathalpha}{letters}{"02}

\let\Phi\undefined\newcommand{\Phi}{\Phiit}
\let\Sigma\undefined\newcommand{\Sigma}{\Sigmait}
\let\Psi\undefined\newcommand{\Psi}{\Psiit}

\newenvironment{myitemize}{\begin{itemize}[itemsep=-0.05cm, leftmargin=0.55cm, topsep=0.1cm]}{\end{itemize}}

\title{Geometry and 2-Hilbert Space\\for Nonassociative Magnetic Translations}

\author{Severin Bunk, Lukas Müller and Richard J. Szabo}

\begin{document}

\begin{flushright}
\small
\textsf{Hamburger Beiträge zur Mathematik Nr.\,728}\\
{\sf ZMP--HH/18--9} \\
{\sf EMPG--18--09} \\
\end{flushright}

\vspace{10mm}

\begin{center}
	\textbf{\LARGE{Geometry and 2-Hilbert Space\\[0.2cm]for Nonassociative Magnetic Translations}}\\
	\vspace{1cm}
	{\large Severin Bunk$^{a}$}, \ \ {\large Lukas Müller$^{b}$} \ \ and \ \ {\large Richard J. Szabo$^{b}$}

\vspace{5mm}

{\em $^a$ Fachbereich Mathematik, Bereich Algebra und Zahlentheorie\\
Universit\"at Hamburg\\
Bundesstra\ss e 55, D\,--\,20146 Hamburg, Germany}\\
Email: {\tt  severin.bunk@uni-hamburg.de\ }
\\[7pt]
{\em $^b$ Department of Mathematics\\
Heriot-Watt University\\
Colin Maclaurin Building, Riccarton, Edinburgh EH14 4AS, U.K.}\\
and {\em Maxwell Institute for Mathematical Sciences, Edinburgh, U.K.}\\
and {\em The Higgs Centre for Theoretical Physics, Edinburgh, U.K.}\\
Email: {\tt lm78@hw.ac.uk and r.j.szabo@hw.ac.uk\ }
\end{center}

\vspace{1cm}

\begin{abstract}
\noindent
We suggest a geometric approach to quantisation of the twisted Poisson
structure underlying the dynamics of charged particles in fields of
generic smooth distributions of magnetic charge, and dually of closed strings
in locally non-geometric flux backgrounds, which naturally allows for representations of nonassociative magnetic translation operators.
We show how one can use the 2-Hilbert space of sections of a bundle
gerbe in a putative framework for canonical quantisation.
We define a parallel transport on bundle gerbes on $\R^d$ and show that it naturally furnishes weak projective 2-representations of the translation group on this 2-Hilbert space.
We obtain a notion of covariant derivative on a bundle gerbe and a novel perspective on the fake curvature condition.
\end{abstract}

\newpage

\tableofcontents

\bigskip

\section{Introduction and summary}

In this paper we consider the quantisation of a \emph{twisted magnetic
  Poisson structure} which is defined in the following way. We work
with the $d$-dimensional real vector space $M=\R^d$ for some $d\in\N$, which we call
`configuration space', and consider its dual `momentum space' $M^*$
with the evaluation pairing denoted by $\langle-,-\rangle:M^*\times M\to\R$. The `phase space' $\frM=T^*M=M\times M^*$ is naturally a symplectic space with the canonical symplectic form $\sigma_0(X,Y):=\langle p,y\rangle - \langle q,x\rangle$, for any two vectors $X=(x,p)$ and $Y=(y,q)$ of $\frM$; we write $x=\sum_{i=1}^d\, x^i\,e_i$ and $p=\sum_{i=1}^d\, p_i\, e^i$ where $e_i$ are the standard basis vectors of $\R^d$ and $e^i$ are their duals, $\langle e^i,e_j\rangle=\delta^i{}_j$. By a `magnetic field' we shall generically mean any 2-form $\rho\in\Omega^2(M)$ on $M$ whose components $\rho_{ij}(x)$ have suitable smoothness properties, and with it we can deform the canonical symplectic structure to an almost symplectic form
\begin{equation}\label{eq:sigmarho}
\sigma_\rho = \sigma_0 - \pi^*\rho \ ,
\end{equation}
where $\pi \colon \frM \longrightarrow M$ denotes the projection onto the base space.
The inverse $\vartheta_\rho=\sigma_\rho^{-1}\in\Gamma(\frM,\bigwedge^2T\frM)$ defines a bivector
\begin{equation}\label{eq:varthetarho}
\vartheta_\rho = \bigg(\begin{matrix}
0 & \One_d \\ -\One_d & -\rho
\end{matrix}\bigg)
\end{equation}
and `twisted magnetic Poisson brackets'
\begin{equation}\label{eq:twistedbrackets}
\{f,g\}_\rho:= \vartheta_\rho(\dd f\wedge\dd g)
\end{equation}
for smooth functions $f,g\in C^\infty(\frM,\C)$. For the coordinate functions $x^i(x,p)=x^i$ and $p_i(x,p)=p_i$ on $\frM$, we have the relations
\begin{equation}\label{Eq: Magnetic commutation relations}
\{x^i,x^j\}_\rho=0 \ , \qquad \{x^i, p_j\}_\rho = \delta^i{}_j \qquad \mbox{and} \qquad \{p_i,p_j\}_\rho = -\rho_{ij}(x) \ .
\end{equation}
Computation of the Schouten bracket in this case,
\begin{equation}\label{eq:Schoutenbracket}
[\vartheta_\rho,\vartheta_\rho]_{\rm S} =
\mbox{$\bigwedge^3$}\vartheta_\rho^\sharp(\dd\sigma_\rho) \ ,
\end{equation}
reveals that $\vartheta_\rho$ defines an $H$-twisted Poisson structure on $\frM$ with the 3-form
\begin{equation}
H := \dd\rho
\end{equation}
on configuration space, called the `magnetic charge'; it determines the Jacobiators
\begin{equation}\label{eq:Jacobiators}
\{f,g,h\}_\rho:=[\vartheta_\rho,\vartheta_\rho]_{\rm S}(\dd f\wedge\dd g\wedge\dd h)
\end{equation}
of the twisted magnetic Poisson brackets \eqref{eq:twistedbrackets}. On coordinate functions the violation of the Jacobi identity is seen through the
possibly non-vanishing Jacobiators
\begin{equation}\label{eq:can associator with B-field}
\{p_i,p_j,p_k\}_\rho = -H_{ijk}(x) \ .
\end{equation}

The twisted magnetic Poisson structure is central to certain
applications to physics. For $d=3$ it governs the motion of a charged
particle in a magnetic field $\bfB= \sum_{i=1}^3\, B^i(x) \, e_i$ on $M=\R^3$ by taking
\begin{equation}
\rho_{ij}(x) = \sum_{k=1}^3\, e\, \varepsilon_{ijk}\, B^k(x) \ ,
\end{equation}
where $e\in\R$ is the electric charge and $\varepsilon$ is the
Levi-Civita symbol. Canonical quantisation of the twisted magnetic Poisson
structure means applying the correspondence principle of quantum
mechanics to linearly associate operators $\Oa_f$ to phase space functions
$f(X)$ such that the brackets \eqref{Eq: Magnetic commutation
  relations} of coordinate functions map to the commutation relations
\begin{equation}
\label{eq:can comm with B-field}
		[\Oa_{x^i},\Oa_{x^j} ] = 0\ , \qquad
		[\Oa_{x^i},\Oa_{p_j} ] = \iu\, \hbar\, \delta^i{}_j \,
                \One \qquad \mbox{and} \qquad 
		[\Oa_{p_i},\Oa_{ p_j} ]= - \iu\,\hbar\, \rho_{ij}(\Oa_x) \ ,
\end{equation}
with a deformation parameter $\hbar\in\R$.
In physics one thus says that a magnetic field $\bfB$ leads to a noncommutative momentum space.
In particular, by the second relation in \eqref{eq:can comm with
  B-field} the operators
\begin{equation}
\Pa_v = \exp\Big(\frac\iu\hbar \, \Oa_{\langle p,v\rangle} \Big)
\end{equation}
implement translations by vectors $v \in \R^3_\rmt$
in the translation group $\R^3_\rmt$ of $M$,
\begin{equation}
\Pa_v^{-1}\, \Oa_{x^i}\, \Pa_v
= \Oa_{x^i+v^i} \ ,
\end{equation}
and by the third relation they do not commute; we refer to $\Pa_v$ as
`magnetic translations'. Note that the map $f\mapsto\Oa_f$ does not
generally send twisted Poisson brackets to commutators, since for
functions $f,g\in C^\infty(\frM,\C)$ one has
\begin{equation}
[\Oa_f,\Oa_g]=\iu\, \hbar\, \Oa_{\{f,g\}_\rho} + O(\hbar^2) \ ,
\end{equation}
where the order $\hbar^2$ corrections are non-zero only when $f$ and $g$ are at most quadratic in $X=(x,p)$.

In the classical Maxwell theory of electromagnetism, the magnetic
field $\bfB$ is free of sources, i.e.~${\rm div}(\bfB)=
\sum_{i=1}^3\,\nabla_{e_i}B^i = 0$, and hence $\rho$ is a closed
2-form, so that the bivector $\vartheta_\rho$ defines a Poisson structure in this instance (equivalently $\sigma_\rho$ is a symplectic form). In this case the magnetic translation operators can be represented geometrically as parallel transport on a hermitean line bundle with connection $(L,\nabla^L)$ on $M$ (see e.g.~\cite{Hannabuss:2017ion,Soloviev:Dirac_Monopole_and_Kontsevich}), where the curvature 2-form of $\nabla^L$ is given by $F_{\nabla^L} = \rho$.
Then the quantum Hilbert space of the charged particle is the space $\Hscr =
\mathrm{L}^2(M,L)$ of square-integrable global sections of $L$ (with
respect to the Lebesgue measure). This geometric approach is reviewed
in Section~\ref{Sec: Case without sources}.

On the other hand, Dirac's semi-classical modification of the Maxwell
theory allows for distributions of magnetic sources.
For the field of a single Dirac monopole of magnetic charge $g\in\R$
located at the origin,
\begin{align}
	B^i(x)= g \, \frac{x^i}{|x|^3} \ ,
\end{align}
which is defined on the configuration space $M=\R^3 \setminus \{0\}$,
the line bundle $L\to M$ is non-trivial and only exists if the Dirac
charge quantisation condition $\frac{2\,e\,g}\hbar\in\Z$ is satisfied~\cite{Wu-Yang}.

When considering generic magnetic fields with sources of magnetic charge, the momentum operators $\Oa_{{p}_1}, \Oa_{{p}_2}, \Oa_{{p}_3}$ fail to associate and from the Jacobiator \eqref{eq:can associator with B-field} we find
\begin{equation}\label{eq:3-can associator with B-field}
	[\Oa_{{p}_1}, \Oa_{{p}_2}, \Oa_{{p}_3}] = \hbar^2\, e\, {\rm
          div}(\bfB)(\Oa_x) \ .
\end{equation}
That is, the failure of associativity of the operators $\Oa_{{p}_i}$ is proportional to the magnetic charge density, whence for generic magnetic fields the operators $\Oa_{{p}_i}$ are part of a nonassociative algebra.
The corresponding magnetic translations $\Pa_u, \Pa_v, \Pa_w$ for vectors $u,v,w \in \R^3_\rmt$ no longer associate either, with the failure of associativity controlled by a 3-cocycle on the translation group $\R^3_\rmt$ that takes values in the $\sfU(1)$-valued functions on $M$.
Consequently, in fields of generic smooth distributions of magnetic
charge $H=e\, {\rm div}(\bfB)\, \dd x^1\wedge\dd x^2\wedge\dd x^3$, the algebra of observables of a charged particle can no longer be represented on a Hilbert space.
The effect on the geometric side is the breakdown of the description in terms of a line bundle on $M$. This was already observed long ago by~\cite{Jackiw:1984rd}.
In the case of a collection of isolated Dirac monopoles on $\R^3$ one can circumvent these issues by excluding the locations of the monopoles from $\R^3$~\cite{Wu-Yang}, but this is not feasible when one wishes to describe smooth distributions of magnetic charge.

In addition to this conceptual interest in quantum mechanics, for general dimension $d$ the twisted magnetic Poisson structure also plays a role in certain non-geometric string theory compactifications, see e.g.~\cite{Blumenhagen:2010hj,Lust:2010iy,Blumenhagen:2011ph,MSS:NonGeo_Fluxes_and_Hopf_twist_Def}, after applying a `magnetic duality transformation'. For this, note first that the transformation $(x,p)\mapsto(p,-x)$ of order~4 preserves the canonical symplectic form $\sigma_0$, and if we further trade the 2-form $\rho$ on $M$ for a 2-form $\beta$ on $M^*$ then we obtain the bivector
\begin{equation}
\vartheta^\vee_\beta = \bigg(\begin{matrix}
-\beta & \One_d \\ -\One_d & 0
\end{matrix}\bigg) \ ,
\end{equation}
which leads to the twisted Poisson brackets of coordinate functions
\begin{equation}
\{x^i,x^j\}_\beta = -\beta^{ij}(p) \ , \qquad \{x^i,p_j\}_\beta = \delta^i{}_j \qquad \mbox{and} \qquad \{p_i,p_j\}_\beta = 0 \ .
\end{equation}
Now the twisting is by a 3-form
\begin{equation}
R := \dd\beta
\end{equation}
on momentum space, called an `$R$-flux', and in particular the non-vanishing Jacobiators on coordinate functions are given by
\begin{equation}
\{x^i,x^j,x^k\}_\beta = -R^{ijk}(p) \ .
\end{equation}
In this case one speaks of closed strings propagating in a
noncommutative and nonassociative configuration space upon
quantisation, which is interpreted as saying that the $R$-flux
background is `locally non-geometric'. Our results in this paper 
shed light on what should substitute for canonical quantisation of
locally non-geometric closed strings.

Even though the Hilbert space framework is unavailable, the observables still form a well-defined algebra in the case of generic magnetic field $\rho$, as originally studied in~\cite{Gunaydin:1985ur} and more recently within an algebraic approach to nonassociative quantum mechanics in~\cite{Bojowald:2014oea}. Thus far there exist two approaches to the full quantisation of twisted magnetic Poisson structures. The original approach of~\cite{MSS:NonGeo_Fluxes_and_Hopf_twist_Def} is based on deformation quantisation and it provides explicit nonassociative star products, which have been developed from other perspectives and applied to nonassociative quantum mechanics in~\cite{Bakas:2013jwa,Mylonas:2013jha,Barnes:2014ksa,KS:G_2_and_quantisation}. Another approach is to embed the $2d$-dimensional twisted Poisson manifold $(\frM,\vartheta_\rho)$ into a $4d$-dimensional symplectic manifold by extending the technique of symplectic realisation from Poisson geometry~\cite{KS:Symplectic_realisation}; in this approach geometric quantisation can be used and the standard operator-state methods of canonical quantum mechanics employed, but at the cost of trading the nonassociativity for the introduction of spurious auxiliary degrees of freedom which cannot be eliminated.

In this paper we provide a third perspective on nonassociativity that is more along the lines of an operator-state framework in canonical quantisation, and which avoids the introduction of extra variables.
It was suggested by~\cite{Szabo:Magnetic_monopoles_and_NAG} that a suitable geometric framework to handle nonassociativity analogously to the source-free case $H=0$ would be to replace line bundles by bundle gerbes with connections.
Bundle gerbes are a categorified analogue of line bundles whose curvatures are 3-forms rather than 2-forms, and these curvature 3-forms can model the magnetic charges $H=\dd\rho$.
We show that for generic smooth distributions of magnetic charge on $M=\R^d$ the nonassociative translation operators $\Pa_v$ are realised naturally as the parallel transport functors of a suitably chosen bundle gerbe $\Ia_\rho$ on $M$.
This suggests a canonical geometric representation of the
nonassociative algebra of observables, where the state space is the 2-Hilbert space $\Gamma(M, \Ia_\rho)$ of global sections of $\Ia_\rho$; this is in complete analogy with the line bundle description in the case of vanishing magnetic charge. In particular, this gives a precise meaning to the formal manipulations of~\cite{Jackiw:1984rd} which subsumed quantities on which the nonassociative magnetic translations were represented for generic magnetic charge density $H$: in our case these quantities are realised as sections of a bundle gerbe $\Ia_\rho$ on~$\R^3$.

In Section~\ref{Sec: Case without sources} we start by reviewing the associative  situation with $H=0$, and in particular the realisation of magnetic translations as the parallel transport on a hermitean line bundle with connection. We also recall the magnetic Weyl correspondence which provides the link between geometric quantisation and deformation quantisation via the parallel transport.
In Section~\ref{Sec: Bundle gerbes} we introduce the geometric formalism of bundle gerbes on $\R^d$.
Like line bundles on $\R^d$, bundle gerbes on $\R^d$ can be described very explicitly up to equivalence, and we specialise to that simplified setting.
This allows us to give a very concrete model for the 2-Hilbert space of sections $\Gamma(\R^d, \Ia_\rho)$ of bundle gerbes $\Ia_\rho$ on $\R^d$.
We then provide a definition of projective representations of groups on categories rather than on vector spaces in Section~\ref{sect:weak projective 2-reps}: their failure to be honest representations is measured by group 2-cocycles whose target is a category rather than a module.

The parallel transport $\Pa$ on a given bundle gerbe $\Ia_\rho$ on $M=\R^d$ is defined in Section~\ref{sect:weak projective rep of translations}.
We show that $\Pa$ induces what we call a `weak projective 2-representation' of the translation group $\R^d_\rmt$ of $M$ on the 2-Hilbert space of global sections $\Gamma(M, \Ia_\rho)$; the higher 2-cocycle that twists this 2-representation is exactly the 3-cocycle discussed originally in~\cite{Jackiw:1984rd} and derived precisely by~\cite{MSS:NonGeo_Fluxes_and_Hopf_twist_Def,Bakas:2013jwa,Mylonas:2013jha,Szabo:Magnetic_monopoles_and_NAG} in the framework of deformation quantisation. What is currently lacking is a categorical version of the magnetic Weyl correspondence, defined in terms of the parallel transport functors $\Pa$, that bridges the approaches to quantisation of the twisted magnetic Poisson structure through the 2-Hilbert space $\Gamma(M, \Ia_\rho)$ and deformation quantisation.
In the final Section~\ref{sect:Covariant derivatives of sections} we take a step towards an infinitesimal version of the magnetic translations, which can be interpreted as momentum operators on $\Gamma(M, \Ia_\rho)$ and could lead to a higher version of the magnetic Weyl correspondence (and ultimately Hamiltonian dynamics) in this setting.
We are able to give a notion of a tangent category to $\Gamma(M, \Ia_\rho)$ as well as a covariant derivative of sections of $\Ia_\rho$.
We show that the covariant derivative ties in with the parallel transport functors and that it very naturally gives rise to the fake curvature condition from higher gauge theory.

\subsection*{Conventions and notation}
\label{sect:conventions and notation}

For the convenience of the reader, we summarise here our notation and conventions to be used throughout this paper.

\begin{myitemize}
	\item Most of the geometry in this paper will take place on the smooth manifold $M=\R^d$.
	The manifold $\R^d$ carries a simply transitive action of the translation group in $d$ dimensions, which we denote by $\R^d_\rmt$ in order to distinguish it from the smooth manifold $\R^d$.
	The global vector field canonically associated to a translation vector $v \in \R^d_\rmt$ will be denoted by $\hat{v} \in \Gamma(\R^d, T\R^d)$.
	That is, $\hat{v}|_{x} = v$ for every $x \in \R^d$.
	Often we will tacitly use the canonical identification of the tangent space $T_x\R^d$ with~$\R^d$.
	
	\item Let $v_1, \ldots, v_m \in \R^d_\rmt$ be translation
          vectors for some integer $m \geq 1$.
	We denote by $\triangle^m(x;v_1,\dots ,v_m)$ the $m$-simplex spanned by $ x - \sum_{i=1}^m\, v_i,  x- \sum_{i=2}^{m}\,v_i, \dots , x-v_m, x$ (see Figure \ref{Fig: Sketch D3} for the case $m=3$).
	Concretely we set
	\begin{align}
		\triangle^m(x;v_1,\dots ,v_m) &= \bigg\{ x - \sum_{i=1}^m\, v_i + \sum_{i=1}^m\, t_i \ \sum_{j=1}^i \, v_j \in \R^d \ \bigg| \ t_i \in [0,1]\,, \ \sum_{i=1}^m\, t_i \leq 1 \bigg\}
		\\[4pt]
		&= \bigg\{ x - \sum_{i=1}^m\, \bigg( 1 - \sum_{j=i}^m\, t_j \bigg)\, v_i \in \R^d \ \bigg| \ t_i \in [0,1]\,, \ \sum_{i=1}^m\, t_i \leq 1 \bigg\} \ .
	\end{align}
Let $\triangle^m = \triangle^m(\sum_{i=1}^m\, e_i; e_1, \ldots, e_m) \subset \R^m$ denote the standard geometric $m$-simplex, where $e_1, \ldots, e_m$ denotes the standard basis of $\R^m$.
	There is a canonical smooth map
	\begin{equation}
		\delta^m(x;v_1, \ldots, v_m) \colon \triangle^m \to \R^d\ ,
		\quad
		(t_1, \ldots, t_m) \mapsto x - \sum_{i=1}^m\, \bigg( 1 - \sum_{j=i}^m\, t_j \bigg)\, v_i\ .
	\end{equation}
	For an $m$-form $\eta \in \Omega^m(\R^d)$ we introduce the shorthand notation
	\begin{equation}
		\int_{\triangle^m(x;v_1,\dots ,v_m)}\, \eta \coloneqq \int_{\triangle^m}\, \big( \delta^m(x;v_1, \ldots, v_m) \big)^* \eta
	\end{equation}
	for the integral of $\eta$ over the $m$-simplex in $\R^d$ based at $x - \sum_{i = 1}^m\, v_i$ and spanned by the vectors $v_1, \ldots, v_m$.
	
For $m=3$, the boundary of $\triangle^3(x;e_1,e_2,e_3)$ with the induced orientation decomposes as
	\begin{equation}
	\label{Eq: Boundary delta3}
	\begin{split}
		\partial \triangle^3(x;e_1,e_2,e_3) =\
		&\triangle^2(x;e_2,e_3)\ \cup\ \triangle^2(x;e_1,e_2+e_3)
		\\
		&\ \cup\ \overline{\triangle^2(x;e_1+e_2,e_3)}\ \cup\ \overline{\triangle^2(x-e_3;e_1,e_2)}\ ,
	\end{split}
	\end{equation}
	where an overline denotes orientation reversal.
	
\begin{figure}[tb]
\centering
\begin{overpic}[scale=0.5]%
{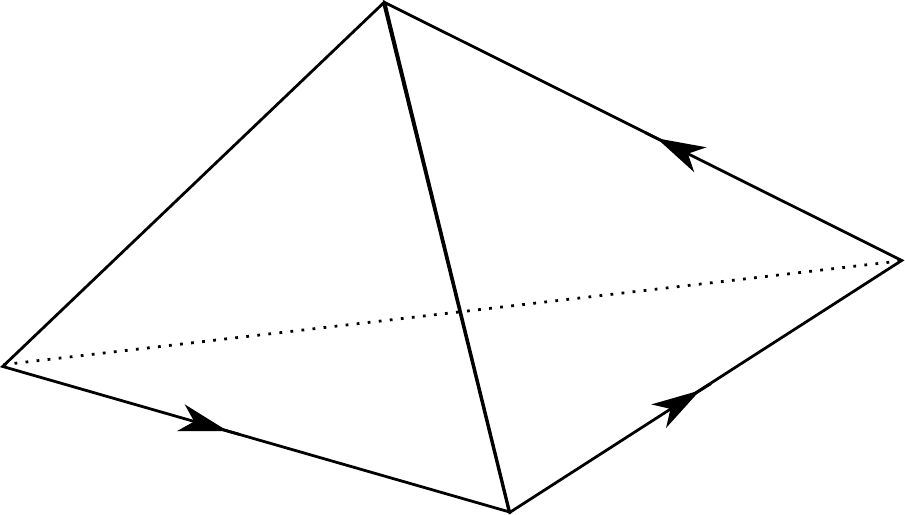}
\put(-48,14){\scriptsize$x-w-v-u$}
\put(60,-10){\scriptsize$x-v-u$}
\put(132,35){\scriptsize$x-u$}
\put(52,75){\scriptsize$x$}
\end{overpic}
\bigskip
\caption{\small The 3-simplex $\triangle^3(x;w,v,u)$.}
\label{Fig: Sketch D3}
\end{figure}
	
	\item Let $\sfG$ be a group.
	We denote by $\rmB \sfG$ the category with a single object $*$ and one morphism for every element $g \in \sfG$.
	Composition of these morphisms is given by the group multiplication in $\sfG$.
	
	\item If $\Ca$ is a category and $k \in \N$ we define $\underline{\Ca}_{\,k}$ to be the (strict) $k$-category obtained by adding only trivial morphisms in degrees $1 < l \leq k$ to $\Ca$.
	
	\item If $\Ca$ is a symmetric monoidal category and $k \in \N$, we let $\rmB^k \Ca$ denote the $(k+1)$-category obtained by placing $\Ca$ in degrees $k$ and $k+1$ while having a single object in every other degree.
	
	\item For a category $\Ca$, we let $\mathrm{Aut}(\Ca)$ denote the Picard groupoid that has auto-equivalences of $\Ca$ as objects and natural isomorphisms as morphisms.
	If $\Ca$ is monoidal, we let $\mathrm{Aut}_\otimes(\Ca)$ denote the Picard groupoid that consists of monoidal auto-equivalences and monoidal natural isomorphisms.
	
	\item We let $\mathrm{BiCat}$ denote the 3-category of bicategories, 2-functors, 2-natural transformations and modifications.
	
	\item Given an object $r$ of a symmetric monoidal category $\mathcal{R}$ and a module category $\Ma$ over $\Ra$, we denote by $\ell_r \colon \Ma \longrightarrow \Ma$ the action of $r$ on $\Ma$.
	
	\item Throughout this paper we do not indicate coherence
          morphisms explicitly, such as associators and unitors in
          monoidal categories (or braidings in the symmetric case), in
          order to streamline notation.
	Due to Mac~Lane's Coherence Theorem~\cite{ML:Categories} they can always be reinstated in an essentially unique way.
\end{myitemize}

\section{Associative magnetic translations}
\label{Sec: Case without sources}

If the magnetic field $\rho\in\Omega^2(M)$ is defined everywhere on
$M=\R^d$ and is closed, $\dd \rho = 0$, then there exists a connection on the trivial complex line bundle \smash{$L = M \times \C \xrightarrow{\text{pr}_{M}} M$} with curvature $\rho$; the connection can be described by a globally defined 1-form $A= \sum _{i=1}^dA_i\dd x^i\in \Omega^1(M)$ satisfying $\dd A = \rho $.
Quantising the Poisson structure $\vartheta_\rho$ produces the Hilbert space $\Hscr=\mathrm{L}^2(M, L)$ of square-integrable sections of $L$ (defined with respect to the Lebesgue measure). The coordinate functions $x^i$ and $p_i$ on $\frM$ correspond to self-adjoint operators 
\begin{equation}
	O_{x^i} \colon {\rm dom}(O_{x^i}) \longrightarrow \Hscr\ , \quad 
	(O_{x^i}\psi)(x) = x^i\, \psi(x)
\end{equation}
and 
\begin{equation}
	O_{{p}_i} \colon {\rm dom}(O_{{p}_i}) \longrightarrow \Hscr\ , \quad
	(O_{p_i}\psi)(x) = \left(-\iu\,\hbar\, \frac{\partial}{\partial x^i} + A_i \right)\psi(x) \ ,
\end{equation}
for all $x\in M$, each defined on a dense subspace of $\Hscr$.
One easily verifies the commutation relations \eqref{eq:can comm with B-field} which quantise the magnetic Poisson brackets \eqref{Eq: Magnetic commutation relations}. 

The momentum operator $O_{p_i}$ is nothing but $-\iu\,\hbar$ times the covariant derivative in the direction $e_i$ induced by the connection $\nabla^L = \dd + \frac\iu\hbar\, A$ on $L$; therefore the canonical action of the translation group $\R^d_\rmt$ on $\Hscr$ generated by the momentum operators is given by the parallel transport in the line bundle $L$.
Concretely, for $v \in \R^d_\rmt$ we define the magnetic translation operators
\begin{equation}
	P_v \colon \Hscr \longrightarrow \Hscr\ ,
 	 \quad (P_v \psi)(x) = P^{\nabla^L}_{\triangle^1(x;v)}\, \psi(x-v) \ ,
\end{equation} 
for all $x\in M$.
Here the parallel transport in the line bundle $L$ with connection $\nabla^L$ along the path $t \mapsto x - (1-t)\,v$ for $t \in [0,1]$ is given by
\begin{equation}
	P^{\nabla^L}_{\triangle^1(x;v)} = \exp\bigg(-\frac\iu\hbar\, \int_{\triangle^1(x;v)}\, A\bigg)\ .
\end{equation}
Evaluating the composition of two translation operators on a section $\psi \in \Gamma(M, L)$ using Stokes' Theorem gives 
\begin{equation}\label{Eq: quasi-projective representation}
\begin{aligned}
	\big( P_v P_w \psi \big) (x) & = P^{\nabla^L}_{\triangle^1(x;v)}\, P^{\nabla^L}_{\triangle^1(x-v,w)}\, \big( \psi(x-v-w) \big)
	\\[0.2cm]
	&= \exp \bigg( -\frac\iu\hbar\, \int_{\partial \triangle^2(x;w,v)}\, A \ \bigg)\, (P_{v+w} \psi) (x) 
	\\[0.2cm]
	& = \exp \bigg( -\frac\iu\hbar\, \int_{ \triangle^2(x;w,v)}\, \rho \bigg) \, (P_{v+w} \psi) (x) \ .
\end{aligned}
\end{equation}
This shows that the magnetic translation operators do not respect the group structure of the translation group strictly, but only up to a failure measured by the  $\sfU(1)$-valued function $\omega_{v,w} \in  C^\infty(M, \sfU(1))$ with
\begin{align}
\label{EQ: Definition 2 cocycle}
	\omega_{v,w}(x) \coloneqq \exp \bigg( -\frac\iu\hbar\, \int_{ \triangle^2(x;w,v)}\, \rho \bigg) \ ,
\end{align}
for all $x \in M$ and $v,w \in \R^d_\rmt$.

We will now study the collection of these functions and in particular the dependence of $\omega_{v,w}$ on the vectors $v, w \in \R^d_\rmt$ in more detail.  

\begin{definition}
Let $\sfG$ be a group and let $\sfM$ be a $\sfG$-module, i.e. an abelian group $\sfM$ with a compatible $\sfG$-action $\tau\colon \sfG \times \sfM \to \sfM$, $(g,\psi)\mapsto \tau_g(\psi)$.
A \emph{2-cocycle $\omega$ on $\sfG$ with values in $\sfM$} is a function $\omega \colon \sfG \times \sfG \to \sfM$, $(g,h)\mapsto\omega_{g,h}$ satisfying
\begin{align}
	\tau_g(\omega_{h,k}) \, \omega^{-1}_{g\,h,k} \, \omega_{g,h\,k} \, \omega^{-1}_{g,h} = 1 \ ,
\end{align}  
for all $g,h,k \in \sfG$, and we write $\omega\in\mathsf{C}^2(\sfG,\sfM)$.
\end{definition}

\begin{proposition}
The collection of functions $\{\omega_{v,w}\}_{v,w \in \R^d_\rmt}$
defined in \eqref{EQ: Definition 2 cocycle} defines a
2-cocycle
\begin{equation}
	\omega \in \mathsf{C}^2 \big( \R^d_\rmt,C^\infty(M, \sfU(1)) \big)
\end{equation}
on the translation group $\R^d_\rmt$ with respect to the action 
\begin{align*}
	\tau \colon \R^d_\rmt \times C^\infty  \big( M, \sfU(1) \big) &\longrightarrow C^\infty \big( M, \sfU(1)  \big)\ , \quad
	(v, f)  \longmapsto \tau_{-v}^*f
\end{align*}
of $\R^d_\rmt$ on the abelian group $C^\infty(M, \sfU(1))$, where $\tau_v:M\to M$ is the translation $x\mapsto x+v$ by $v\in\R^d_\rmt$ for all $x\in M$.
\end{proposition}

\begin{proof}
We check the 2-cocycle condition at an arbitrary but fixed point $x \in M$:
\begin{align}
	\omega_{v,w}(x-u)\, \omega_{u+v,w}^{-1}(x)\, \omega_{u,v+w}(x)\, \omega_{v,w}^{-1}(x)
	&= \exp \bigg(  -\frac\iu\hbar\, \int_{\partial \triangle^3(x;u,v,w)}\, \rho \bigg)
	\\[4pt]
	&=\exp \bigg(  -\frac\iu\hbar\, \int_{\triangle^3(x;u,v,w)}\, \dd \rho \bigg)
	\\[4pt]
	&= 1
\end{align}
for all triples of translation vectors $u,v,w \in \R^d_\rmt$.
\end{proof}

Consequently, the parallel transport operators $P^{\nabla^L}$ do not furnish a representation of the translation group $\R^d_\rmt$ on the Hilbert space $\Hscr = \mathrm{L}^2(M,L)$.
Rather, we find the following structure (see also~\cite{Soloviev:Dirac_Monopole_and_Kontsevich}).

\begin{definition}
\label{def: weak projective rep}
Let $\sfG$ be a group, let $\sfA$ be a unital commutative $\C$-algebra
whose group of invertible elements is denoted $\sfA^\times$, and let $\sfM$ be an $\sfA$-module (which is a complex vector space since $\sfA$ is unital).
Assume that $\sfA$ carries a $\sfG$-action $\tau \colon \sfG \times \sfA \to \sfA$ compatible with the algebra operations, i.e.~$\tau_g(f\, f') = \tau_g(f)\, \tau_g(f')$ for all $f,f' \in \sfA$ and $g \in \sfG$.
Let $\omega \in \sfC^2(\sfG, \sfA^\times)$ be an $\sfA^\times$-valued 2-cocycle with respect to the action $\tau$.
A \emph{weak projective representation of $\sfG$ on $\sfM$ twisted by $\omega$} is a map
\begin{equation}
	\rho \colon \sfG \times \sfM \longrightarrow \sfM\ ,
	\quad
	(g,\psi) \mapsto \rho_g(\psi)\ ,
\end{equation}
respecting the group structure of $\sfM$ such that
\begin{equation}
	\rho_g (f\triangleright \psi) = \tau_g(f)\triangleright \rho_g(\psi)
	\qandq
	\rho_h \circ \rho_g (\psi) = \omega_{h,g}\triangleright \rho_{h\,g} (\psi)
\end{equation}
for all $f \in \sfA$, $\psi \in \sfM$ and $g,h \in \sfG$.
If $\sfG$ is a Lie group, and $\sfA$ and $\sfM$ carry smooth structures, we additionally demand that $\tau$ and $\rho$ be smooth maps.
\end{definition}

Hence we have found a weak projective representation of $\R^d_\rmt$ twisted by the 2-cocycle $\omega$ of the translation group on $\sfM = \Hscr$, where the relevant algebra is $\sfA = C_b^\infty(M, \C)$, the algebra of bounded smooth functions, with $\R^d_\rmt$-action $\tau_v(f)(x) \coloneqq f(x-v)$ for all $v \in \R^d_\rmt$ and $f \in C_b^\infty(M, \C)$.
The 2-cocycle $\omega$ is in fact trivial in group cohomology;
it is the coboundary of the $C^\infty(M, \sfU(1))$-valued 1-cochain $\lambda$ given by
\begin{align}
	\lambda_v (x) \coloneqq \exp\bigg( -\frac\iu\hbar\, \int_{\triangle^1(x;v)}\, A \bigg) \ .
\end{align}
This trivialisation corresponds to passing from the noncommutative kinematical momentum operators $O_{p_i}$ to the commuting gauge-variant canonical momentum operators $O_{p_i}-A_i(O_x)$~\cite{Jackiw:1984rd}, which at the level of coordinate functions on $\frM$ sends the symplectic form \eqref{eq:sigmarho} to the canonical form $\sigma_0$. 

In the special case of a constant magnetic field, where the 2-form
$\rho=\tilde\rho \in \Omega^2(M)$ is constant, the 2-cocycle $\omega$ factors through $\sfU(1)$ regarded as a trivial $\R^d_\rmt$-module.
That is, denoting by $\imath \colon \sfU(1) \to C^\infty(M, \sfU(1))$ the embedding of $\sfU(1)$ into $C^\infty(M, \sfU(1))$ as constant functions, there exists a 2-cocycle $\tilde{\omega} \colon \R^d_\rmt\times \R^d_\rmt\longrightarrow \sfU(1)$ which makes the diagram
\begin{equation}
\begin{tikzcd}
	\R^d_\rmt \times \R^d_\rmt  \ar[r, "\omega"] \ar[d, dashed, "\tilde{\omega}",swap]  & C^\infty(M, \sfU(1))\\
  \sfU(1) \ar[ru,hookrightarrow, "\imath"'] 
\end{tikzcd} 
\end{equation}
commute; explicitly $\tilde\omega_{v,w}=\exp(-\frac\iu{2\hbar} \, \tilde\rho(v,w))$. In this case the weak projective representation of $\R^d_\rmt$ on $\Hscr$ reduces to an honest projective representation.

The parallel transport operators $P^{\nabla^L}$ used above also
provide the bridge between canonical quantisation and (strict)
deformation quantisation, i.e. the Weyl correspondence between
operators and symbols, see
e.g.~\cite{Mantoiu:2005tb,Soloviev:Dirac_Monopole_and_Kontsevich}. For
completeness and for later reference, let us outline the
correspondence. Let $\Sscr(M)$ denote the space of Schwartz functions
on $M=\R^d$ endowed with its Fr\'echet topology and $\Sscr'(M)$ its
(topological) dual space of tempered distributions.

We begin by recalling from~\cite{Mantoiu:2005tb} the `magnetic Weyl system' which is the family of unitary operators $\{W(X)\}_{X\in\frM}$ on $\Hscr$ defined by
\begin{equation}\label{eq:Weylsystem}
\begin{aligned}
W(X):\Hscr\to\Hscr \ , \quad \big(W(x,p)\psi\big)(y) &=
\e^{\frac{\iu\,\hbar}2\, \langle p,x\rangle}\, \e^{-\ii\langle
  p,y\rangle} \, (P_x\psi)(y) \\[4pt]
&= \e^{\frac{\iu\,\hbar}2\, \langle
  p,x\rangle}\, \e^{-\ii\langle p,y\rangle} \,
P^{\nabla^L}_{\triangle^1(y;x)}\, \psi(y-x) \ .
\end{aligned}
\end{equation}
The `magnetic Weyl quantisation map' then proceeds in analogy with the usual Weyl correspondence: For any $f\in \Sscr(\frM)$ we define a bounded operator $O_f:\Sscr(M)\to\Sscr'(M)$ by
\begin{equation}
f\mapsto O_f =\frac1{(2\pi)^{d/2}} \, \int_\frM\, \bigg(\frac1{(2\pi)^{d/2}} \, \int_\frM\, \e^{\iu\,\sigma_0(X,Y)}\, f(Y) \ \dd Y\bigg)\, W(X)\ \dd X \ ,
\end{equation}
and following the standard terminology we refer to the phase space
function $f=:\Wscr(O_f)$ as the magnetic Weyl symbol of the operator $O_f$;
in particular, one readily checks that the magnetic translation of a
symbol $f\in\Sscr(\frM)$ is realised as the conjugation action
\begin{equation}
O_{P_xf} = W(X)^{-1} \, O_f \, W(X)
\end{equation}
in the magnetic Weyl system~\cite{Iftimie2010}. 

This construction induces a `magnetic Moyal-Weyl star product'
$\star_\rho:\Sscr(\frM)\otimes_\C\Sscr(\frM)\to\Sscr(\frM)$ such that
\begin{equation}
O_{f\star_\rho g} = O_f \, O_g \ .
\end{equation}
Explicitly, it is given by a twisted convolution product defined by the oscillatory integrals
\begin{equation}
(f\star_\rho g)(X) = \frac1{(\pi\,\hbar)^{2d}}\, \int_\frM \ \int_\frM\ \e^{-\frac{2\,\iu}\hbar\, \sigma_0(Y,Z)}\, \omega_{x+y-z,x-y+z}(x-y-z)\, f(X-Y)\, g(X-Z)\ \dd Y\ \dd Z
\end{equation}
where $\omega$ is the 2-cocycle defined in \eqref{EQ: Definition 2
  cocycle}. This star product turns the Schwartz space $\Sscr(\frM)$
into a noncommutative associative $\C$-algebra and has similar properties
to the usual Moyal-Weyl star product which is recovered for $\rho=0$
(i.e. $\omega=1$). In particular, if $\rho=\tilde\rho\in\Omega^2(M)$ is constant,
we can replace $\omega$ by the $\sfU(1)$-valued 2-cocycle $\tilde\omega$ and
write the magnetic Moyal-Weyl star product in terms of the
corresponding symplectic
form \eqref{eq:sigmarho} as
\begin{equation}\label{eq:MoyalWeylconstant}
(f\star_{\tilde\rho} g)(X) = \frac1{(\pi\,\hbar)^{2d}}\, \int_\frM \ \int_\frM\
\e^{-\frac{2\,\iu}\hbar\, \sigma_{\tilde\rho}(Y,Z)}\, f(X-Y)\, g(X-Z)\ \dd Y\
\dd Z \ .
\end{equation}

In the following we aim to generalise these constructions to the case
of magnetic fields on $M$ with sources, whereby $H=\dd\rho\neq0$.

\section{Bundle gerbes on $\R^d$}
\label{Sec: Bundle gerbes}

In this section we review the 2-category of hermitean line bundle gerbes with connection on $\R^d$.
We will usually say bundle gerbe when we mean a hermitean line bundle gerbe with connection.
Bundle gerbes can be understood as a categorification of hermitean line bundles with connection.\footnote{For the construction of the 2-Hilbert space of sections it is important to work with categorified line bundles instead of principal bundles, since it allows the existence of non-invertible morphisms.}
They were introduced in~\cite{Murray--Bundle_gerbes} as a geometric structure that describes the second differential cohomology of the base manifold, in a way analogous to how hermitean line bundles with connection describe the first differential cohomology.
In particular, bundle gerbes are geometric objects that give rise to field strength 3-forms.
By now there is a well-developed theory of the 2-category of bundle gerbes~\cite{Waldorf--More_morphisms,Bunk--Thesis}; a less technical introduction to the framework of bundle gerbes with an eye towards applications in string theory and M-theory can be found in~\cite{Bunk-Szabo--Fluxes_brbs_2Hspaces}.
We point out that bundle gerbes are not the only available geometric model for degree-two differential cocycles.
Other prominent models are, for instance, based on sheaves of groupoids~\cite{Brylinski--Loop_Spaces_and_GeoQuan} and predate the discovery of bundle gerbes.
Models that generalise to other differential cohomology theories have been developed in~\cite{BNV--DiffCoho,HS--DiffCoho}.
For us, however, bundle gerbes are the most convenient model, for they allow for a straightforward interpretation as a categorification of line bundles and their sections.
Thus, they can be related rather directly to quantum mechanics at a conceptual level (see~\cite{BSS--HGeoQuan,Bunk-Szabo--Fluxes_brbs_2Hspaces}), and this conceptual relation to quantum mechanics underlies our treatment of bundle gerbes in the rest of this paper.

\subsection{2-category of trivial bundle gerbes on $\R^d$}

Bundle gerbes without connection on a manifold $M$ are classified by the third integer cohomology $H^3(M,\Z)$ of $M$.
Thus up to equivalence we only need to understand topologically trivial bundle gerbes on $M=\R^d$.
We will concentrate on this situation here.
A trivial hermitean line bundle on $\R^d$ with a non-trivial connection is completely described by its connection 1-form $A \in \Omega^1(\R^d)$.
A morphism of trivial line bundles is equivalently a function $f \in C^\infty(\R^d, \C)$; the morphism is unitary if $f$ is $\sfU(1)$-valued, and if the source and target line bundles carry connections $A_0$ and $A_1$, respectively, the morphism $f$ is parallel if and only if it satisfies $\iu\, A_1\, f = \iu\, f\, A_0 - \dd f$, i.e. precisely if it is a gauge transformation.
Composition of morphisms amounts to multiplication of functions.
Connection 1-forms $A$ and morphisms $f$ assemble into a category $\HLBdl^\nabla_\triv(\R^d)$ of trivial hermitean line bundles with (possibly non-trivial) connection on $\R^d$.
This category is symmetric monoidal under the monoidal product
\begin{equation}
	\big(A_0 \overset{f}{\to} A_1\big) \otimes \big(A'_0 \overset{f'}{\to} A'_1\big) = A_0 + A'_0 \overset{f\, f'}{\to} A_1 + A'_1
\end{equation}
for $A_j, A'_j \in \Omega^1(\R^d)$, $j = 0,1$ and $f, f' \in C^\infty(\R^d, \C)$.
This monoidal product is nothing but the tensor product of line bundles with connection restricted to trivial line bundles; its monoidal unit is $A=0$.

The central idea to categorifying line bundles is to replace scalars by vector spaces.
As we are interested only in topologically trivial bundle gerbes here, we do not need to consider any transition functions.
Under the above paradigm, generic $\C$-valued functions are categorified to hermitean vector bundles.
Finally, in order to obtain curvature 3-forms, the connection on a bundle gerbe must be given by a 2-form.
We let $\frh(n) \subset \Mat(n \times n, \C)$ denote the Lie algebra of
hermitean $n \times n$ matrices.

\begin{definition}
\label{def:BGrb_triv}
The \emph{2-category $\BGrb_\triv^\nabla(\R^d)$ of trivial bundle gerbes on
  $\R^d$} is given as follows:
\begin{myitemize}
	\item An object is given by a 2-form $\rho \in \Omega^2(\R^d)$.
	We also denote the object corresponding to $\rho$ by
        $\Ia_\rho$.
	The 3-form $H \coloneqq \dd \rho \in \Omega^3(\R^d)$ is the \emph{curvature of $\Ia_\rho$}.
	
	\item A 1-morphism $\Ia_{\rho_0} \to \Ia_{\rho_1}$ is a 1-form $\eta \in \Omega^1(\R^d, \frh(n))$ for some $n \in \N_0$.
	This is to be thought of as the connection 1-form of as a hermitean connection on the trivial rank~$n$ hermitean vector bundle $E_\eta = \R^d \times \C^n \to \R^d$.
	A 1-morphism $\eta \colon \Ia_{\rho_0} \to \Ia_{\rho_1}$ is
        \emph{fake flat} if its \emph{fake curvature} $F_\eta -
        (\rho_1 - \rho_0)\cdot\One_n$ vanishes, where $F_\eta = \dd \eta +
        \frac{\iu}{2} \, [\eta, \eta]$ is the field strength of the connection $\dd + \iu\, \eta$ on the bundle $E_\eta$.
	
	\item If the 1-forms $\eta \in \Omega^1(\R^d, \frh(n))$ and $\eta' \in \Omega^1(\R^d, \frh(n'))$ for $n,n' \in \N_0$ are 1-morphisms from $\Ia_{\rho_0}$ to $\Ia_{\rho_1}$, a 2-morphism $\eta \Longrightarrow \eta'$ is a matrix-valued function $\psi \in C^\infty(\R^d, \Mat(n \times n', \C))$.
	A 2-morphism $\psi$ is \emph{unitary} if $\psi(x)$ is unitary for all $x \in \R^d$, and it is \emph{parallel} if it satisfies $\iu\, \eta'\, \psi = \iu\, \psi\, \eta - \dd \psi$.
	A 2-morphism $\psi$ is to be thought of as a morphism $\psi \colon E_\eta \to E_{\eta'}$ of hermitean vector bundles with connection.
	
	\item Composition of 1-morphisms $\eta \colon \Ia_{\rho_0} \to \Ia_{\rho_1}$ and $\eta' \colon \Ia_{\rho_1} \to \Ia_{\rho_2}$ is given by $(\eta, \eta') \mapsto \eta' \otimes \One + \One \otimes \eta$.
	Horizontal composition of 2-morphisms is given by $(\psi, \phi) \mapsto \phi \otimes \One + \One \otimes \psi$.
	
	\item Vertical composition of 2-morphisms is given by pointwise matrix multiplication, i.e. it reads as $(\phi, \psi) \mapsto \psi\, \phi$.
\end{myitemize}
\end{definition}

This is a simplified version of the general 2-category of bundle
gerbes on $\R^d$ in the same way that $\HLBdl^\nabla_\triv(\R^d)$ is a simplification of the category of hermitean line bundles with connection on $\R^d$.
However, as pointed out above, since $\R^d$ is homotopically trivial, the 2-category $\BGrb^\nabla_\triv(\R^d)$ is in fact equivalent to the general 2-category of bundle gerbes on $\R^d$ and will, therefore, be perfectly sufficient for us to use in the present setting.

\begin{remark}
\label{rmk:BGrb_triv makes sense on every M}
The 2-category defined as above makes sense if we
replace $\R^d$ by any manifold $M$.
However, the equivalence to the full 2-category of bundle gerbes does not hold on generic base manifolds $M$.
Nevertheless, the full 2-category of bundle gerbes on $M$ can be constructed from $\BGrb^\nabla_\triv(M)$ by closing this under descent along surjective submersions or, equivalently, good open coverings~\cite{NS--Equivariance_in_higher_geometry}.
\end{remark}

The 2-category $\BGrb^\nabla_\triv(\R^d)$ carries a symmetric monoidal structure, denoted $\otimes$, which is defined as follows:
given two objects $\Ia_\rho, \Ia_{\rho'} \in \BGrb^\nabla_\triv(\R^d)$ we set
\begin{equation}
	\Ia_\rho \otimes \Ia_{\rho'} \coloneqq \Ia_{\rho + \rho'}\ .
\end{equation}
Given 1-morphisms $\eta, \nu \colon \Ia_{\rho_0} \to \Ia_{\rho_1}$ and $\eta', \nu' \colon \Ia_{\rho'_0} \to \Ia_{\rho'_1}$, as well as 2-morphisms $\psi \colon \eta \Longrightarrow \nu$ and $\psi' \colon \eta' \Longrightarrow \nu'$, we set
\begin{equation}
	\eta \otimes \eta' \coloneqq \eta \otimes \One + \One \otimes \eta'
	\qandq
	(\psi \otimes \psi')(x) \coloneqq \psi(x) \otimes \psi'(x)
	\ ,
\end{equation}
for all $x\in\R^d$.
The unit of the symmetric monoidal structure $\otimes$ is the \emph{trivial bundle gerbe with connection $\Ia_0$}.

There is an additional symmetric monoidal structure $\oplus$ on the category of morphisms $\Ia_{\rho_0} \to \Ia_{\rho_1}$.
For 1-morphisms $\eta, \eta', \nu, \nu' \colon \Ia_{\rho_0} \to \Ia_{\rho_1}$, and 2-morphisms $\phi \colon \eta \Longrightarrow \eta'$ and $\psi \colon \nu \Longrightarrow \nu'$, it reads as
\begin{equation}
	(\eta \oplus \eta')|_{x} \coloneqq \eta|_{x} \oplus \eta'|_{x}
	\qandq
	(\phi \oplus \psi)(x) \coloneqq \phi(x) \oplus \psi(x)
	\ ,
\end{equation}
for all $x\in\R^d$.
Note that $\otimes$ is monoidal with respect to $\oplus$ in each argument -- in other words, on 1-morphisms and 2-morphisms $\otimes$ distributes over $\oplus$.
In particular, the category $\BGrb^\nabla_\triv(\R^d)(\Ia_0, \Ia_0)$ of endomorphisms of the trivial bundle gerbe is a categorified ring, or a \emph{rig category}.
Moreover, there exists an action of this rig category on every other morphism category $\BGrb^\nabla_\triv(\R^d)(\Ia_{\rho_0}, \Ia_{\rho_1})$ induced by the tensor product $\otimes$ of (trivial) bundle gerbes.
This turns the morphism categories in $\BGrb^\nabla_\triv(\R^d)$ into \emph{rig module categories} over the rig category $\BGrb^\nabla_\triv(\R^d)(\Ia_0, \Ia_0)$.
Note that
\begin{equation}
	\big( \BGrb^\nabla_\triv(\R^d)(\Ia_0, \Ia_0), \otimes, \oplus \big) \cong \big( \HVBdl_\triv(\R^d), \otimes, \oplus \big) \ ,
\end{equation}
i.e. the rig category of endomorphisms of the trivial bundle gerbe $\Ia_0$ is equivalent to the rig category of trivial hermitean vector bundles with possibly non-trivial connection on $\R^d$.
For gerbes $\Ia_{\rho_j}$ with $\rho_j \neq 0$ for $j = 0,1$ there is still an equivalence
\begin{equation}
	\big( \BGrb^\nabla_\triv(\R^d)(\Ia_{\rho_0}, \Ia_{\rho_1}), \oplus \big) \cong \big( \HVBdl_\triv(\R^d), \oplus \big)
\end{equation}
of rig modules over $\HVBdl_\triv(\R^d)$, but the morphism categories $\BGrb^\nabla_\triv(\R^d)(\Ia_{\rho_0}, \Ia_{\rho_1})$ are not closed under the monoidal product $\otimes$ on the 2-category $\BGrb^\nabla_\triv(\R^d)$ so that they do not form rig categories themselves unless $\rho_0 = \rho_1 = 0$.

\subsection{2-Hilbert space of sections}
\label{sect:2Hspace of sections of triv BGrbs}

If $I = M \times \C$ is the trivial hermitean line bundle on a manifold $M$ and $L$ is an arbitrary hermitean line bundle on $M$, there are canonical isomorphisms
\begin{equation}
	\Gamma(M,L) \cong \HLBdl(M)(I,L)
	\qandq
	C^\infty(M,\C) \cong \Gamma(M,I) \cong \HLBdl(M)(I,I)\ .
\end{equation}
This motivates the following definition.

\begin{definition}
Let $\rho \in \Omega^2(\R^d)$ and let $\Ia_\rho$ be the corresponding trivial bundle gerbe on $\R^d$.
The \emph{category of global sections of $\Ia_\rho$} is
\begin{equation}
	\Gamma(\R^d, \Ia_\rho) \coloneqq \BGrb^\nabla_\triv(\R^d)(\Ia_0, \Ia_\rho)\ .
\end{equation}
For $\rho = 0$ we call $\Gamma(\R^d, \Ia_0) \cong \HVBdl_\triv(\R^d)$ the \emph{rig category of higher functions on $\R^d$}.
\end{definition}

\begin{remark}
\label{eq:sections are equiv to Hilb}
The equivalence of $\Gamma(\R^d, \Ia_0)$-module categories
\begin{equation}
	\big( \Gamma(\R^d, \Ia_\rho), \oplus \big) \cong \big( \HVBdl_\triv(\R^d), \oplus \big)
\end{equation}
holds for any 2-form $\rho \in \Omega^2(\R^d)$, just like the space of sections of a line bundle $L$ is independent of the choice of a connection on $L$.
\end{remark}

The hermitean metric on a hermitean line bundle $L$ can be encoded in a non-degenerate positive-definite $C^\infty(M,\C)$-sesquilinear morphism
\begin{equation}
	h_L \colon \Gamma(M,L) \times \Gamma(M,L) \to C^\infty(M,\C)\ .
\end{equation}
Given a trivial bundle gerbe $\Ia_\rho$ on $\R^d$, sections $\eta, \eta', \nu, \nu' \colon \Ia_0 \to \Ia_\rho$, and 2-morphisms $\phi \colon \eta \Longrightarrow \eta'$ and $\psi \colon \nu \Longrightarrow \nu'$, we define a functor
\begin{equation}
[-,-] \colon \Gamma(\R^d, \Ia_\rho)^\opp \times \Gamma(\R^d, \Ia_\rho)
\to \HVBdl_\triv(\R^d)
\end{equation}
by
\begin{equation}
[\eta, \nu] \coloneqq \nu \otimes \One + \One \otimes (-\eta^\sft)
\qquad \mbox{and} \qquad [\phi, \psi] \coloneqq \psi \otimes \phi^\sft\ .
\end{equation}
This is the explicit expression on trivial vector bundles of forming homomorphism bundles.
We view the bifunctor $[-,-]$ as the higher analogue of a hermitean bundle metric.
This idea was developed in the general setting of bundle gerbes in~\cite{BSS--HGeoQuan,Bunk--Thesis}; a less technical treatment can be found in~\cite{Bunk-Szabo--Fluxes_brbs_2Hspaces}.
Now let $\Gamma_\rmpar \colon \HVBdl \to \Hilb$ denote the functor
that takes parallel global sections\footnote{A section of a hermitian vector bundle is \emph{parallel} if it is annihilated by the covariant derivative.}, where $\Hilb$ is the 
category of finite-dimensional complex Hilbert spaces. Then assigning to a pair $\eta, \nu \colon \Ia_0 \to \Ia_\rho$ of sections of a bundle gerbe on $\R^d$ the finite-dimensional Hilbert space%
\footnote{This space is finite-dimensional since the dimension of the space of parallel sections of a vector bundle on a path-connected manifold is bounded from above by the rank of the vector bundle.}
\begin{equation}
\label{eq:inner product on cats of sections}
	\< \eta, \nu \> \coloneqq \Gamma_\rmpar \big( \R^d, [\eta, \nu] \big)
\end{equation}
provides a categorified hermitean inner product bifunctor
\begin{equation}
	\<-,-\> \colon \Gamma(\R^d, \Ia_\rho)^\opp \times \Gamma(\R^d, \Ia_\rho) \to \Hilb
\end{equation}
on the category of smooth sections of $\Ia_\rho$.

Finally, the rig category $(\Hilb, \otimes, \oplus)$ of
finite-dimensional Hilbert spaces naturally embeds into the rig
category $(\HVBdl_\triv(\R^d), \otimes, \oplus)$ by assigning to a
finite-dimensional Hilbert space $V$ the trivial bundle with the
trivial connection $\imath(V) \coloneqq (\R^d \times V \to \R^d, \dd)$ and to a linear map $\psi \in \Hilb(V,W)$ the constant bundle morphism $\imath(\psi)$ with $\imath(\psi)(x) \coloneqq \psi$ for all $x \in \R^d$.
We view the functor $\imath$ as the higher analogue of how scalars give rise to constant functions on a manifold.
In this way $\Gamma(\R^d, \Ia_\rho)$ becomes a rig module category over $\Hilb$.

\begin{definition}
\label{def:2Hspace of sections}
Let $\Ia_\rho \in \BGrb^\nabla_\triv(\R^d)$ be a trivial bundle gerbe with connection on $\R^d$.
The $\Hilb$-module category $\Gamma(\R^d, \Ia_\rho)$ together with the $\Hilb$-sesquilinear bifunctor $\<-,-\>$ is called the \emph{2-Hilbert space of sections of $\Ia_\rho$}.
\end{definition}

\section{Weak projective 2-representations}
\label{sect:weak projective 2-reps}

In this section we introduce a higher version of Definition~\ref{def: weak projective rep} by replacing both the algebra $\sfA$ and the module $\sfM$ by analogous categorified objects, while leaving the structure of the group $\sfG$ unchanged.
Our main interest will concern actions of $\sfG$ on the 2-Hilbert space of sections of a bundle gerbe.
In Section~\ref{Sec: Bundle gerbes} we have seen that the 2-Hilbert spaces arising in this way can be considered as module
categories over the category of hermitean vector bundles with
connection, which in turn can be regarded as an algebra over $\Hilb$, the
category of finite-dimensional complex Hilbert spaces.
In order to approach the final definition, we first define higher 2-cocycles and projective 2-representations of groups on module categories over a symmetric monoidal category.
We then generalise these definitions to allow for weak representations by functors which do not preserve the module structure strictly.
Our discussion follows and generalises ideas
from~\cite{Fiorenza:2014kga, nCat=Cohomology} in the language of 3-categories.\footnote{See~\cite{GPS95} for the corresponding definitions.}
A detailed discussion of projective representations in the language of
2-categories can be found in~\cite[Section~3.4]{Parity}.
Throughout we make the structure and relations of higher (weak) 2-cocycles and
(weak) projective 2-representations very explicit.

\subsection{Category-valued 2-cocycles}

To set the stage we recall the definition of a group action on a
category.

\begin{definition}\label{Def: Group action on categories}
Let $\sfG$ be a group and $\mathcal{C}$ a category. An \emph{action of $\sfG$ on $\mathcal{C}$} is a 2-functor 
\begin{align}
\Theta \colon \underline{\rmB\sfG}_{\,2} \to \rmB \text{Aut}(\mathcal{C}) \ .
\end{align} 
\end{definition} 

\begin{remark}
Unpacking this compact definition we obtain the following data:
\begin{myitemize}
\item 
A functor $\Theta_g \colon \mathcal{C}\to
\mathcal{C}$ for every $g\in \sfG$ ;

\item
Natural isomorphisms $\Pi_{g,h}\colon
\Theta_g \circ \Theta_h \Longrightarrow \Theta_{g\,h}$ for every $g,h \in \sfG$ ; 

\item 
A natural isomorphism $\theta \colon \Theta_1 \Longrightarrow \id_\mathcal{C}$ ;
\end{myitemize}
satisfying the relations
\begin{align}
	\Pi_{g,h\, k} \bullet \big(\Pi_{g,h} \circ \id_{\Theta_k}\big) =  \Pi_{g,h\,k} \bullet \big(\id_{\Theta_g} \circ \Pi_{h,k}\big) \qquad \text{and} \qquad \Pi_{1,g} =\Pi_{g,1} = \theta
\end{align}
for all $g,h,k \in \sfG$, where $\bullet$ denotes the vertical composition in the 2-category $\rmB\mathrm{Aut}(\Ca)$.
\end{remark}

We now fix a symmetric monoidal category $(\mathcal{R},\otimes,1)$ and let $\text{Pic} (\mathcal{R})$ denote the Picard groupoid of $\mathcal{R}$, which is the maximal subgroupoid of $\Ra$ on the objects that are invertible with respect to the monoidal product.

\begin{definition}
Let $\sfG$ be a group.
A \emph{higher 2-cocycle on $\sfG$ with values in a symmetric monoidal category $\mathcal{R}$} is a $3$-functor
\begin{equation}
	\omega \colon \underline{\rmB \sfG}_{\,3} \longrightarrow \rmB^{2} \text{Pic} (\mathcal{R})\ .
\end{equation}
\end{definition}
\begin{remark}\label{Rem: Def. h. 2-Cocycles}
Spelling out the definition of a higher 2-cocycle as a 3-functor we obtain the following structure, which is similar to~\cite[Remark~3.8]{HFixPoint}:
\begin{myitemize}
	\item An object $\imath \in \text{Pic} (\mathcal{R})$\footnote{The element $\imath$ encodes the coherence 2-isomorphism $\omega(1)\Longrightarrow \id $. } ;
	
	\item An
          object $\chi_{g,h}\in \text{Pic} (\mathcal{R})$ for all pairs $g, h \in \sfG$ of group elements ;
	
	\item An isomorphism $\omega_{g,h,k} \colon \chi_{g\,h,k} \otimes \chi_{g,h}\longrightarrow \chi_{g,h\,k} \otimes \chi_{h,k}$ in $\Ra$ for all $g, h, k \in \sfG$ 
	(compare~\cite[Eq.~(3.7)]{HFixPoint}) ;
	
	\item Isomorphisms $\gamma_g \colon \imath \otimes
          \chi_{1,g}\longrightarrow 1$ and $\delta_g \colon 1
          \longrightarrow \chi_{g,1}\otimes \imath $ for every element
          $g\in \sfG$ (compare~\cite[Eqs.~(3.8) and~(3.9)]{HFixPoint}).
\end{myitemize}
All other structure is trivial by the properties of the 3-categories involved. We have to replace modifications in \cite{HFixPoint} with 3-morphisms in our case.
This data is subject to the following conditions:
\begin{myitemize}
	\item
	The diagram
	\begin{equation}
	\begin{tikzcd}
		\chi_{g\,h\,k,l} \otimes \chi_{g,h\,k} \otimes \chi_{h,k} \ar[rr,"\omega_{g,h \,k, l}\otimes \id"] & & \chi_{g,h\,k\,l} \otimes \chi_{h\,k,l} \otimes \chi_{h,k} \ar[dr,"\id \otimes \omega_{h,k,l}"] & \\
		& & & \chi_{g,h\,k\,l} \otimes \chi_{h,k\,l} \otimes \chi_{k,l} \\
		\chi_{g\,h\,k,l} \otimes \chi_{g\,h,k} \otimes \chi_{g,h} \ar[uu, "\id \otimes \omega_{g,h,k}"] \ar[rr, "\omega_{g\, h,k,l}\otimes \id",swap] & & \chi_{g\,h,k\,l} \otimes \chi_{k,l} \otimes \chi_{g,h} \ar[ru, "\omega_{g,h,k\, l}\otimes \id",swap] &
	\end{tikzcd}
	\end{equation}
	commutes~\cite[Axiom~(HTA1)]{GPS95};
	
	\item The identity $(\id_{\chi_{g,h}} \otimes \gamma_h)\circ \omega_{g,1,h}
          \circ (\id_{\chi_{g,h}} \otimes \delta_h)=\id_{\chi_{g,h}}$ holds in
          $\Ra$~\cite[Axiom~(HTA2)]{GPS95};
\end{myitemize} 
for all $g,h,k,l\in\sfG$. We complement this by the simplifying normalisation conditions
\begin{align}
	\imath&=\chi_{g,1}= \chi_{1,g}=1 \ ,
	\\[4pt]
	\delta_g&=\gamma_g= \id \ ,
	\\[4pt]
	\omega_{g,h,1}&= \omega_{g,1,h}=\omega_{1,g,h}=\id\ ,
\end{align}
for all $g,h\in \sfG$.
\end{remark}

There is a natural notion of morphisms between higher 2-cocycles given by lax 3-natural transformations (that are strict with respect to identities). 
The problem with this definition is that the standard definition of a 3-natural transformation, appearing for example in \cite{GPS95}, is that of an op-lax 3-natural transformation in the terminology of \cite{Johnson-Freyd_Scheimbauer}.
To our knowledge, the definition of a lax 3-natural transformation is
not spelled out explicitly in the literature, and we refrain from doing so in this paper.
Instead we give a more concrete definition, which we arrived at by reversing the directions of arrows in the definition corresponding to op-lax 3-natural transformations.

\begin{definition}
\label{Coboundary}
Let $(\chi, \omega)$ and $(\chi', \omega')$ be higher 2-cocycles on a
group $\sfG$ with values in a symmetric monoidal category $\Ra$.
A \emph{higher 2-coboundary} $(\lambda , \Lambda) \colon (\chi, \omega) \longrightarrow (\chi', \omega')$ is given by the data of:
\begin{myitemize}
	\item 
	An object $\lambda_g \in \mathcal{R}$ for every group element
        $g \in \sfG$ ; 
	
	\item 
	Isomorphisms $\Lambda _{g,h} \colon \chi_{g,h}\otimes
        \lambda_g \otimes \lambda_h \longrightarrow
        \lambda_{g\,h}\otimes \chi'_{g,h}$ for all $g,h \in \sfG$ ;
\end{myitemize}
satisfying the normalisation conditions $\lambda_1= 1$, $\Lambda _{1,g} = \Lambda _{g,1} = \id$ and the coherence condition given by the commutativity of the diagram
\begin{equation}
\begin{tikzcd}[column sep=1.5cm, row sep=1cm]
	\chi_{g\,h,k}\, \chi_{g,h}\, \lambda_g\, \lambda_h\, \lambda_k \ar[r,"\omega_{g,h,k}"] \ar[d, "\Lambda _{g,h}"'] & \chi_{g,h\,k}\, \chi_{h,k}\, \lambda_g\, \lambda_h\, \lambda_k \ar[r,"\Lambda _{h,k}"] & \chi_{g,h\,k}\, \chi'_{h,k}\, \lambda_{g}\, \lambda_{h\,k} \ar[d, "\Lambda _{g,h\,k}"]
	\\
	\chi_{g\,h,k}\, \chi'_{g,h}\, \lambda_{g\,h}\, \lambda_k \ar[r,"\Lambda _{g\,h,k}"'] &\chi'_{g\,h,k}\, \chi'_{g,h}\, \lambda_{g\,h\,k} \ar[r,"\omega'_{g,h,k}"'] & \chi'_{g,h\,k}\, \chi'_{h,k}\, \lambda_{g\,h\,k} 
\end{tikzcd}
\end{equation}
for all $g,h,k \in \sfG$.
Here we do not display monoidal products, braidings or identities for brevity.
If $(\lambda , \Lambda) \colon (\chi, \omega) \longrightarrow (\chi', \omega')$ and $(\lambda' , \Lambda') \colon (\chi', \omega') \longrightarrow (\chi'', \omega'')$ are two higher 2-coboundaries, their composition reads as $(\lambda \otimes \lambda', \Lambda' \circ \Lambda)$.
\end{definition}

\begin{remark}
There is a natural definition of morphisms between higher 2-coboundaries and of morphisms between these morphisms which we do not spell out explicitly.
Formally, this stems from the fact that our definitions can be seen to take place in the 4-category of 3-categories.
\end{remark}

\begin{definition}
The \emph{second group cohomology of $\sfG$ with values in the Picard
groupoid of a symmetric monoidal category $\Ra$} is the abelian group
$H^2(\sfG,\text{Pic}(\mathcal{R}))$ obtained as the quotient of the
collection of higher $\Ra$-valued 2-cocycles $(\chi, \omega)$ on
$\sfG$ by the equivalence relation $(\chi, \omega) \sim (\chi', \omega')$ if and only if there exists a higher 2-coboundary $(\lambda, \Lambda) \colon (\chi, \omega) \to (\chi', \omega')$.
\end{definition}

A symmetric monoidal functor $F\colon \mathcal{R} \longrightarrow \mathcal{R}'$ induces a natural map
\begin{equation}
F_* \colon H^2 \big( \sfG,\text{Pic}(\mathcal{R}) \big)
\longrightarrow H^2 \big( \sfG,\text{Pic} (\mathcal{R}') \big) \ ,
	\quad
	[\chi, \omega] \mapsto [F(\chi), F (\omega) ]\ .
\end{equation}

\subsection{Projective 2-representations on module categories}

We can embed $\rmB^2 \text{Pic} (\mathcal{R})$ into the 3-category
$\text{BiCat}$ of bicategories by sending the only object to the
bicategory $\Ra\text{-mod}$ of $\mathcal{R}$-module categories, module functors and natural transformations, the only 1-morphism to the identity, objects $r \in \mathcal{R}$ (regarded as 2-morphisms) to the 2-natural transformation $\ell_r \colon \id_{\mathcal{R}\text{-mod}} \Longrightarrow \id_{\mathcal{R}\text{-mod}}$ with components $\ell_{r|\Ma} \colon \Ma \longrightarrow \Ma$ given by the action of $r$ on $\Ma$ for every $\mathcal{R}$-module category $\Ma$, and a morphism $f \colon r \longrightarrow r'$ to the induced modification $ f\colon \ell_r \Rrightarrow \ell_{r'}$. We denote the composition of a higher 2-cocycle $\omega$ with this embedding again by $\omega$. 
Using this embedding we can give an elegant definition of a projective 2-representation $\rho$ twisted by a higher 2-cocycle $\omega$ with values in $\mathcal{R}$ as a (lax) 3-natural transformation
\begin{equation}
\begin{tikzcd}
 & \, & \\
\underline{\rmB \sfG}_{\,3} \ar[rr, bend left, "\omega"] \ar[rr, bend right, "1", swap] & & \text{BiCat} \\
 &\, \ar[uu, Rightarrow, shorten <=18 ,shorten >=18, "\, \rho",swap] & 
\end{tikzcd}
\end{equation}
where $1$ is the constant 3-functor at the terminal 2-category with only one object, 1-morphism and 2-morphism.
Again, the problem with this definition is that the standard
definition of a 3-natural transformation, appearing for example
in~\cite{GPS95}, is that of an op-lax 3-natural transformation in the
terminology of~\cite{Johnson-Freyd_Scheimbauer}. Hence we give a more
concrete definition.

\begin{definition}
\label{Rem: Def higher projectiv representation}
Let $\sfG$ be a group, $\Ra$ a symmetric monoidal category and $(\chi, \omega)$ a higher 2-cocycle on $\sfG$ with values in $\Ra$.
A \emph{projective 2-representation of $\sfG$ over $\Ra$ twisted by $(\chi, \omega)$} consists of:
\begin{myitemize}
	\item An $\mathcal{R}$-module category $\Ca$ ;
	
	\item An $\Ra$-module functor $\Theta_g \colon \Ca
          \longrightarrow \Ca $ for every group element $g \in \sfG$ ;
	
	\item A natural isomorphism $\Pi_{g,h} \colon \Theta_g \circ
          \Theta_h \Longrightarrow \ell_{\chi_{g,h}}  \circ
          \Theta_{g\,h}$ for each pair $g, h \in \sfG$ ;
	
	\item A natural isomorphism $\theta \colon \Theta_1
          \Longrightarrow \id_\Ca $ . 
\end{myitemize}
These data are subject to the coherence conditions given by  $\Pi_{1,g} = \Pi_{g,1} = \theta$ and the commutativity of the diagram
\begin{equation}
\label{EQ: Diagram definition Projective}
\begin{tikzcd}[column sep=1.25cm, row sep=1cm]
	\Theta_g \circ \ell_{\chi_{h,k}} \circ \Theta_{h\,k} \ar[r,Rightarrow] & \ell_{\chi_{h,k}} \circ \Theta_g  \circ \Theta_{h\,k} \ar[r,Rightarrow, "\Pi_{g,h\,k}"] & \ell_{\chi_{g,h\,k}} \circ \ell_{\chi_{h,k}}  \circ \Theta_{g\,h\,k}
	\\
	\Theta_g \circ  \Theta_h \circ \Theta_k \ar[r,Rightarrow,"\Pi_{g,h}"'] \ar[u,Rightarrow, "\Theta_g(\Pi_{h,k})"]& \ell_{\chi_{g,h}}\circ  \Theta_{g\,h}  \circ \Theta_{k} \ar[r,Rightarrow,"\Pi_{g\,h,k}"'] & \ell_{\chi_{g\,h,k}} \circ \ell_{\chi_{g,h}} \circ \Theta_{g\,h\,k} \ar[u,Rightarrow,swap, "\ell_{\omega_{g,h,k}}"]
\end{tikzcd}
\end{equation}
for all $g,h,k \in \sfG$.
\\
We again impose the simplifying normalisation conditions $\Theta_{1} = \id_\Ca$ and $\theta = \id_{\id_\Ca}$.
\end{definition}

The unlabelled isomorphism in the diagram~\eqref{EQ: Diagram
  definition Projective} arises from the property that $\Theta_g$ is
an $\Ra$-module functor -- it commutes with all functors of the form
$\ell_r \colon \Ca \to \Ca$ for $r \in \Ra$ up to coherent isomorphism.
If the higher 2-cocycle $(\chi, \omega)$ is trivial, i.e. if $\chi_{g,h} = 1_\Ca$ is the monoidal unit in $\Ca$ and $\omega_{g,h,k} = \id$ for all $g,h,k \in \sfG$, the data of a projective 2-representation of $\sfG$ on $\Ca$ reduces to that of an honest representation of $\sfG$ on $\Ca$ by $\Ra$-module functors $\Theta_g$.
These still do not respect the group multiplication strictly, but only
up to coherent isomorphism $\Pi_{g,h}$.

\subsection{Weak projective 2-representations}

In Section \ref{Sec: Case without sources} we saw that if the line bundle with connection $(L,\nabla^L)$ on $\R^d$ has non-constant curvature, the parallel transport on $L$ only induces a \emph{weak} projective representation of the translation group on the space of sections of $L$.
Motivated by this observation we proceed to categorify Definition~\ref{def: weak projective rep}.
The difference between a projective representation of $\sfG$ on an
$\sfA$-module $\sfM$ and a weak projective representation is that in
the latter case the algebra $\sfA$ carries a non-trivial $\sfG$-action $\tau$ itself, and the $\sfG$-action on $\sfM$ is by weak module maps relative to $\tau$ (see Definition~\ref{def: weak projective rep}).
In the categorified formalism, this affects the unlabelled isomorphism
in the diagram~\eqref{EQ: Diagram definition Projective}, which was a
consequence of the property that the representing functors $\Theta_g$ are $\Ra$-module functors.
We thus have to introduce a more general definition of these functors.

\begin{definition}
Given two $\mathcal{R}$-module categories $\Ca$ and $\Ca'$, a
\emph{twisted $ \mathcal{R}$-module functor $\Ca \longrightarrow \Ca'$} 
is a pair of functors $(F \colon \Ca \longrightarrow \Ca' , \tau
\colon \mathcal{R} \longrightarrow \mathcal{R})$, where $\tau$ is
symmetric monoidal, together with natural isomorphisms $\eta_{r,c}
\colon F(r \otimes c) \longrightarrow \tau(r) \otimes F(c)$ for all
objects $r\in\Ra$ and $c\in\Ca$ satisfying the usual coherence conditions. 
\end{definition}

Using this definition we can introduce the notion of a higher weak
2-cocycle and a weak projective 2-representation generalising the
analogous notions of Section~\ref{Sec: Case without sources}. 

\begin{definition}
An \emph{action of a group $\sfG$ on a symmetric monoidal category $\mathcal{R}$} is a 2-functor
\begin{align*}
\tau \colon \underline{\rmB \sfG}_{\,2} \longrightarrow \rmB\text{Aut}_\otimes( \mathcal{R}) \ ,
\end{align*} 
where $\text{Aut}_\otimes$ is the groupoid of monoidal autofunctors and monoidal natural isomorphisms between them.
\end{definition}

\begin{definition}\label{def:higher weak 2-cocycle}
Let $\tau$ be an action of a group $\sfG$ on a symmetric monoidal
category $\Ra$.
A \emph{higher weak 2-cocycle on $\sfG$ with values in $\mathcal{R}$ twisted by $\tau$} amounts to giving:
\begin{myitemize}
	\item An object $\chi_{g,h}$ of $\Pic(\Ra)$ for every pair of
          elements $g,h \in \sfG$ ;

	\item An isomorphism $\omega_{g,h,k} \colon \chi_{g\,h,k} \otimes \chi_{g,h}\longrightarrow \chi_{g,h\,k} \otimes \tau(g)[\chi_{h,k}]$  for every triple $g,h,k \in \sfG$ ;
\end{myitemize}
such that $\chi_{g,1} = \chi_{1,g} = 1$ for all $g \in \sfG$, and, supressing isomorphisms corresponding to symmetric monoidal functors and $\tau$, the diagram
\begin{equation}\label{Eq: 3-cocycle condition}
\small
\begin{tikzcd}
 	\chi_{g\,h\,k,l} {\otimes} \chi_{g,h\,k} {\otimes} \tau(g)[\chi_{h,k}] \ar[rr,"\omega_{g,h\, k, l}{\otimes} \id"]  & & \chi_{g,h\,k\,l} {\otimes} \tau(g) [\chi_{h\,k,l}] {\otimes} \tau(g)[\chi_{h,k}] \ar[dr,"{\id {\otimes} \tau(g)[\omega_{h,k,l}]}"] &
 	\\
 	& & & \chi_{g,h\,k\,l} {\otimes} \tau(g)[ \chi_{h,k\,l} {\otimes} \tau(h)[\chi_{k,l}] ]
 	\\
\chi_{g\,h\,k,l} {\otimes} \chi_{g\,h,k} {\otimes} \chi_{g,h} \ar[uu, "{\id {\otimes}  \omega_{g,h,k}}"] \ar[rr, "\omega_{g\, h,k,l}{\otimes} \id"'] & & \chi_{g\,h,k\,l} {\otimes} \tau(g\,h) [\chi_{k,l}] {\otimes} \chi_{g,h} \ar[ru, "\omega_{g,h,k\, l}{\otimes} \id"'] &
\end{tikzcd}
\normalsize
\end{equation}
commutes for all $g,h,k,l \in \sfG$, while if any one entry of $\omega$ is $1$ then $\omega$ is the identity up to structure isomorphisms.
\\
We will sometimes denote this data by the short-hand notation $(\chi, \omega, \tau)$.
\end{definition}

From the point of view of group cohomology, higher weak 2-cocycles are just higher versions of 2-cocycles valued in non-trivial $\sfG$-modules.
The adjective `weak' may thus seem superfluous in this instance, but
we choose this nomenclature to stress their relation to weak projective 2-representations defined below.
We also generalise Definition~\ref{Coboundary}.

\begin{definition}
Let $\sfG$ be a group, $\Ra$ a symmetric monoidal category, $\tau$ a $\sfG$-action on $\Ra$, and let $(\chi, \omega,\tau)$ and $(\chi', \omega',\tau)$ be $\Ra$-valued higher weak 2-cocycles on $\sfG$ twisted by $\tau$.
A \emph{higher weak 2-coboundary} $(\lambda, \Lambda) \colon (\chi,\omega,\tau) \longrightarrow (\chi',\omega',\tau)$ is given by specifying:
\begin{myitemize}
	\item An object $\lambda_g \in \mathcal{R}$ for each $g\in
          \sfG$ ;
	
	\item Isomorphisms $\Lambda_{g,h} \colon \chi_{g,h} \otimes
          \lambda_g \otimes \tau(g)[\lambda_h] \longrightarrow
          \lambda_{g\,h} \otimes \chi'_{g,h}$ for every pair of group
          elements $g,h \in \sfG$ ;
\end{myitemize}
subject to the conditions that $\lambda_1= 1$, that $\Lambda_{1,g} = \Lambda_{g,1} = \id $ for all $g \in \sfG$, and that the diagram
\begin{equation}
\label{EQ: Weak coboundary}
\small
\begin{tikzcd}[column sep=1.5cm, row sep=1cm]
	\chi_{g\,h,k}\, \chi_{g,h}\, \lambda_g\, \tau(g) \big[ \lambda_h\, \tau(h)[ \lambda_k] \big] \ar[r,"\omega_{g,h,k}"] \ar[d, "{\Lambda_{g,h}}"']
	& \chi_{g,h\,k}\, \tau(g)[\chi_{h,k}]\, \lambda_g\, \tau(g) \big[ \lambda_h\, \tau(h)[ \lambda_k] \big] \ar[d,"{\tau(g)[\Lambda_{h,k}]}"]
	\\
	\chi_{g\,h,k}\, \lambda_{g\,h}\, \tau(g\,h)[\lambda_k]\, \chi_{g,h}' \ar[d,"\Lambda_{g\,h,k}"']  & \chi_{g,h\,k}\, \lambda_{g}\, \tau(g)[\chi'_{h,k}\, \lambda_{h\,k}] \ar[d, "\Lambda_{g,h\,k}"]
	\\
	 \chi'_{g\,h,k}\, \chi'_{g,h}\, \lambda_{g\,h\,k} \ar[r,"\omega'_{g,h,k}"'] & \chi'_{g,h\,k}\, \tau(g)[\chi'_{h,k}]\, \lambda_{g\,h\,k} 
\end{tikzcd}
\normalsize
\end{equation}
commutes for all $g,h,k \in \sfG$.
\end{definition}

Adjusting Definition \ref{Rem: Def higher projectiv representation} to the case of higher weak 2-cocycles we arrive at the central concept of this paper.
\begin{definition}
A \emph{weak projective 2-representation of a group $\sfG$ on an $\mathcal{R}$-module category $\Ca$ twisted by a higher weak 2-cocycle $(\chi, \omega, \tau)$} consists of the following data:
\begin{myitemize}
	\item A twisted $\mathcal{R}$-module functor $(\Theta_g,\tau(g)) \colon \Ca\longrightarrow \Ca $ for every $g \in \sfG$ ;
	
	\item A natural isomorphism $\Pi_{g,h} \colon \Theta_g \circ
          \Theta_h \Longrightarrow \ell_{\chi_{g,h}} \circ \Theta_{g\,h}$ for all pairs $g, h \in \sfG$ .
\end{myitemize}
These data are subject to the conditions that $\Theta_{1}=\id$, that
$\Pi_{1,g} = \Pi_{g,1}= \id$ for all $g \in \sfG$, and also that the diagram
\begin{equation}
\label{EQ: Diagram definition weak Projective} 
\begin{tikzcd}[column sep=1.25cm, row sep=1cm]
	\Theta_g \circ \ell_{\chi_{h,k}}\circ  \Theta_{h\,k} \ar[r,Rightarrow]  & \ell_{\tau(g)[\chi_{h,k}]} \circ \Theta_g \circ \Theta_{h\,k}  \ar[r,Rightarrow, "\Pi_{g,h\,k}"] & \ell_{\chi_{g,h\,k}} \circ \ell_{\tau(g)[\chi_{h,k}]} \circ \Theta_{g\,h\,k} 
	\\
	\Theta_g \circ  \Theta_h \circ \Theta_k \ar[r,Rightarrow,"\Pi_{g,h}"'] \ar[u,Rightarrow, "\Theta_g(\Pi_{h,k})"] & \ell_{\chi_{g,h}}\circ \Theta_{g\,h} \circ \Theta_{k} \ar[r,Rightarrow,"\Pi_{g\,h,k}"'] & \ell_{\chi_{g\,h,k}}\circ \ell_{\chi_{g,h}}\circ \Theta_{g\,h\,k} \ar[u,Rightarrow,"\ell_{\omega_{g,h,k}}"']
\end{tikzcd}
\end{equation}
commutes for all $g,h,k \in \sfG$. Here the unlabelled isomorphism
comes from $\Theta_g$ being a twisted $\mathcal{R}$-module functors.
\end{definition}

Given a weak projective 2-representation $(\Theta, \Pi)$ of $\sfG$ on
$\Ca$ twisted by a higher weak 2-cocycle $(\chi, \omega, \tau)$ and a
higher weak 2-coboundary $(\lambda, \Lambda) \colon (\chi, \omega, \tau) \longrightarrow (\chi', \omega', \tau)$ we can define a new weak projective 2-representation $\lambda_*(\Theta, \Pi)$ of $\sfG$ on $\Ca$ twisted by $(\chi', \omega', \tau)$ by setting
\begin{equation}
\label{eq:pushforward of wp2reps}
(\lambda_* \Theta)_g \coloneqq \ell_{\lambda_g} \circ \Theta_g \qquad
\mbox{and} \qquad
	(\lambda_* \Pi)_{g,h} \coloneqq \Lambda_{g,h} \circ \Pi_{g,h}
        \ ,
\end{equation}
for all $g,h\in\sfG$. 

\begin{definition}
\label{Morphisms of weak projective representations}
Let $\Ca$ and $\Ca'$ be $\Ra$-module categories.
A \emph{morphism of weak projective 2-represen{-}tations $(\Ca,\Theta,\Pi) \longrightarrow (\Ca', \Theta',\Pi')$ twisted by the same higher weak 2-cocycle $(\chi, \omega, \tau)$} consists of:
\begin{myitemize}
	\item An $\mathcal{R}$-module functor $\varphi \colon \Ca
          \longrightarrow \Ca'$ ;
	
	\item A natural isomorphism $m_g \colon \varphi \circ \Theta_g
          \Longrightarrow \Theta'_g \circ \varphi$ for every $g \in \sfG$ ;
\end{myitemize}
satisfying $m_1 = \id$ and the coherence condition
\begin{equation}
\begin{tikzcd}[row sep=1cm, column sep=1.cm]
	\Ca \ar[rrrr, bend right=30, "\Theta_{g\,h}"', pos=0.4, ""{name=U, inner sep=0pt, below, pos=.5}] \ar[dd,"\varphi"'] \ar[r, "\Theta_h"] & \Ca \ar[rr, "\Theta_g"] & {} \ar[d, Rightarrow, "\Pi_{g,h}", shorten >= 10, pos=0.3] & \Ca & \Ca \ar[dd,"\varphi"] \ar[l, "\ell_{\chi_{g,h}}"']
	\\
	& & {} & &
	\\
	\Ca' \ar[rrrr, "\Theta'_{g\,h}",swap] & & |[alias=V]| {} & & \Ca'
	\arrow[Rightarrow, from=U, to=V, "m_{g\,h}", shorten <= 5, shorten >= 5]
\end{tikzcd}
	\quad = \quad 
\begin{tikzcd}[row sep=1cm, column sep=1.cm]
	\Ca \ar[dd,"\varphi"'] \ar[r, "\Theta_h"] & \Ca \ar[ld, Rightarrow, "m_h", shorten >= 5, shorten <= 5] \ar[rr, "\Theta_g"] \ar[d,"\varphi"]& \ar[d,Rightarrow, "m_g", shorten >= 5, pos=0.5] & \Ca\ar[d,"\varphi"'] & \Ca \ar[dd,"\varphi"] \ar[l, "\ell_{\chi_{g,h}}"']
	\\
	{} & \Ca'\ar[rr,"\Theta'_g", pos=0.25] & \ar[d, Rightarrow, "\Pi'_{g,h}", shorten >= 5] & \Ca' &\\
	\Ca' \ar[ru, "\Theta_h'"] \ar[rrrr, "\Theta'_{g\,h}",swap] & & {} & & \Ca' \ar[ul,"\ell_{\chi_{g,h}}"]
\end{tikzcd}
\end{equation}
for all $g,h \in \sfG$, where the right-most quadrangle commutes because $\varphi$ is an $\Ra$-module functor.
\end{definition} 

\section{Nonassociative magnetic translations}
\label{sect:weak projective rep of translations}

Let $\rho\in\Omega^2(M)$ be a magnetic field on the configuration space $M=\R^d$. Recall from Section \ref{Sec: Bundle gerbes} that the category
$\Gamma(M,\mathcal{I}_\rho)$ underlying the 2-Hilbert space of
sections of a trivial bundle gerbe $\Ia_\rho$ on $M$ has the following description. Objects are
1-forms $\eta$ on $M$ with values in the Lie algebra $\frh(n)$ of
hermitean $n\times n$ matrices for some $n \in \N_0$.
A morphism $f \colon \eta \longrightarrow \eta'$ from $\eta \in \Omega^1(M,\frh(n))$ to $\eta' \in \Omega^1(M,\frh(n'))$ is a $\Mat(n\times n',\C)$-valued function $f$ on $M$, and we call $f$ parallel if it satisfies
\begin{align}
\label{Condition morphisms}
	\iu\, \eta'\, f = \iu\, f\, \eta - \dd f \ . 
\end{align}  
In this section we will show that the translation group has a natural weak projective 2-representation on the 2-Hilbert space $\Gamma(M,\mathcal{I}_\rho)$. 
Concretely, in the notation of Section~\ref{sect:weak projective
  2-reps} we take $\Ra =  (\HVBdl_\triv(M), \otimes)$ as the
ambient symmetric monoidal category, $\Ma = (\Gamma(M, \Ia_\rho),
\oplus)$ as a module category over $\Ra$, and $\tau(v) = \tau_{-v}^*$ as the
action of $\sfG = \R^d_\rmt$ on $\Ra$ for $v\in\R^d_\rmt$.

\begin{remark}
With these specific choices for $\Ra$, $\Ma$ and $\tau$, we obtain a
direct categorification of the structure of an ordinary weak
projective representation from Definition~\ref{def: weak projective
  rep}. In accordance with the general paradigm of Section~\ref{Sec: Bundle gerbes}, we first replace the ground field $(\C, +, \cdot)$ by the rig category $(\Hilb, \oplus, \otimes)$.
There is a categorified $\Hilb$-algebra structure on $\Aa = (\HVBdl_\triv(M), \otimes, \oplus)$, where the action of $\Hilb$ on $\Aa$ is via mapping Hilbert spaces to trivial vector bundles, as spelled out in Section~\ref{sect:2Hspace of sections of triv BGrbs}.
Finally, $\Ma = (\Gamma(M, \Ia_\rho), \oplus)$ is a module category over $\Aa$.
\end{remark}

In the ensuing concrete calculations we frequently make use of the
fact that the Lie derivative $\mathcal{L}$ and integration are compatible in the following sense:
for every oriented manifold $M$ one has
\begin{align}
	\frac{\dd}{\dd t} \bigg(\int_{\Phi_{\hat v}(t) (V)}\, \eta\bigg)|_{t=0} = \int_{V}\,
  \mathcal{L}_{\hat v} \eta \ ,
\end{align}
where $V$ is an $m$-dimensional submanifold of $M$, $\eta$ is an
$m$-form on $M$, and $\hat v$ is a vector field on $M$ with flow
$\Phi_{\hat v}(t)$ for $t\in [0,1]$.

\subsection{Higher weak 2-cocycle of a magnetic field}
\label{Sec: Def. of the cocycle}

Before introducing magnetic translation operators we define the higher weak
2-cocycle $(\chi,\omega,\tau)$ with values in the symmetric monoidal ($\Hilb$-algebra) category $\mathcal{R} = \HVBdl_\triv (M)$ (cf. Definition~\ref{def:higher weak 2-cocycle}):

\begin{myitemize}
	\item First, we define the action of the translation group $\R^d_\rmt$ on $\HVBdl_\triv(M)$ to be the pullback
	\begin{equation}
	\tau(v) = \tau_{-v}^* \colon \HVBdl_\triv (M) \longrightarrow \HVBdl_\triv (M)\ ,\quad
		\eta \mapsto \tau_{-v}^* \eta \ ,
	\end{equation}
for all $ v \in \R^d_\rmt$.
	
	\item For translation vectors $v,w \in \R^d_\rmt$, the 1-form $\chi_{v,w} \in \Omega^1(M)$ represents the trivial line bundle over $M$ with connection 1-form
	\begin{align}
		\chi_{v,w}|_x(a) =  \frac1\hbar\, \int_{\triangle^2(x;w,v)}\, \iota_{\hat{a}} H \ ,
	\end{align}
for all $a \in T_xM$, where $\iota_{\hat a}$ denotes contraction with
$\hat{a} \in \Gamma(M,TM)$ which is the unique extension of $a$ to a constant vector field on $M$, and $H = \dd \rho$ is the curvature of the bundle gerbe $\Ia_\rho$.
	
	\item Given a triple $u,v,w \in \R^d_\rmt$ of translation vectors we define an isomorphism in $\HLBdl^\nabla_\triv(M)$ via
\begin{equation}
	\omega_{u,v,w} \colon \chi_{u+v,w} \otimes \chi_{u,v} \longrightarrow \chi_{u,v+w}\otimes \tau(u)[\chi_{v,w}]\ ,
	\quad
	\omega_{u,v,w}(x) \coloneqq \exp\bigg(\frac\iu\hbar\,
        \int_{\triangle^3(x;w,v,u)} H \bigg) \ ,
\end{equation}
for all $x \in M$.
\end{myitemize}
We check that this defines a parallel morphism of hermitean vector bundles with connection, i.e. that \eqref{Condition morphisms} is satisfied:
setting $f = \omega_{u,v,w}$, we compute
\begin{equation}
\begin{aligned}
	f^{-1} \, \dd f|_x(a)
	&= \dd \log\exp\bigg(\frac\iu\hbar\, \int_{\triangle^3(-;w,v,u)}\, H\bigg) |_{x}(a)\\[4pt]
	&= \frac\iu\hbar\, \mathcal{L}_{\hat{a}} \bigg( \int_{\triangle^3(-;w,v,u)}\, H\bigg)|_x\\[4pt]
	&=  \frac\iu\hbar\, \int_{\triangle^3(x;w,v,u)} \, \mathcal{L}_{\hat{a}} H\ .
\end{aligned}
\end{equation}
On the other hand, setting $\eta = \chi_{u+v,w} + \chi_{u,v}$ and $\eta' = \chi_{u,v+w} + \tau(u)[\chi_{v,w}]$ in~\eqref{Condition morphisms}, we have 
\begin{align}
	\iu\, (\eta - \eta')|_x(a)
	&= - \frac\iu\hbar\, \int_{\triangle^2(x;v+w,u)}\, \iota_{\hat{a}} H
	- \frac\iu\hbar\, \int_{\triangle^2(x-u,w,v)}\,
          \iota_{\hat{a}} H \\ & \qquad
	+\, \frac\iu\hbar \int_{\triangle^2(x;w,u+v)}\, \iota_{\hat{a}} H
	+ \frac\iu\hbar\, \int_{\triangle^2(x;v,u)}\, \iota_{\hat{a}} H
	\\[4pt]
	&=  \frac\iu\hbar\, \int_{\partial \triangle^3(x;w,v,u)}\,
          \iota_{\hat{a}} H
        \\[4pt]
	&= \frac\iu\hbar\, \int_{\triangle^3(x;w,v,u)}\, \dd\, \iota_{\hat{a}} H
	\\[4pt]
	&=  \frac\iu\hbar\, \int_{\triangle^3(x;w,v,u)}\, \mathcal{L}_{\hat{a}} H\ ,
\end{align}
where we have used \eqref{Eq: Boundary delta3} together with Stokes' Theorem, the Cartan formula for the Lie
derivative $\mathcal{L} = \dd\circ\iota+\iota\circ\dd$, and that $\dd H=0$.
The higher weak 2-cocycle condition~\eqref{Eq: 3-cocycle condition} with respect to the action $\tau$ of $\R^d_\rmt$ on $\HVBdl_\triv(M)$ is satisfied since the curvature 3-form $H$ is closed.

Similarly to the discussion of Section \ref{Sec: Case without
  sources}, we can trivialise the higher weak 2-cocycle
$(\chi,\omega,\tau)$. For this, we denote by $(\chi^0, \omega^0,\tau)$ the trivial higher weak 2-cocycle
twisted by $\tau$ from the above construction, i.e. $\chi^0_{v,w}$ is
the trivial line bundle with trivial connection and $\omega^0_{u,v,w}$
is the identity for all $u,v,w \in \R^d_\rmt$.

\begin{proposition}\label{Prop: Trivialisation of the cocycle}

There is a higher weak 2-coboundary $(\lambda,\Lambda) \colon (\chi, \omega,\tau) \longrightarrow (\chi^0, \omega^0,\tau)$ given by:
\begin{myitemize}
	\item For each $v \in \R^d_\rmt$, the 1-form $\lambda_v$ represents the topologically trivial line bundle with connection 1-form given by
	\begin{equation}
		\lambda_v|_x (a) \coloneqq \frac1\hbar\, \int_{\triangle^1(x;v)}\, \iota_{\hat{a}}\rho\ ,
	\end{equation}
for all $x\in M$ and $a\in T_xM$ ;
	
	\item Given any two translation vectors $v,w \in \R^d_\rmt$, we define an isomorphism of hermitean line bundles
	\begin{equation}
		\Lambda_{v,w} \colon \chi_{v,w} \otimes \lambda_v \otimes \tau_{-v}^*\lambda_w \longrightarrow \lambda_{v+w}\otimes \chi_{v,w}^0 \ ,
		\quad
		\Lambda_{v,w} (x) \coloneqq \exp \bigg(
                \frac\iu\hbar\, \int_{\triangle^2(x;w,v)}\, \rho \bigg)
		\ ,
	\end{equation}
for all $x \in M$.
\end{myitemize}
\end{proposition}

\begin{proof}
We check that $\Lambda_{v,w}$ is parallel for all $x\in M$ and $a\in T_xM$:
\begin{align}
	\dd \bigg( \frac\iu\hbar\, \int_{\triangle^2(-;w,v)}\, \rho \bigg)|_x(a)
	&=\frac\iu\hbar\,\mathcal{L}_{\hat{a}} \bigg( 
          \int_{\triangle^2(-;w,v)}\, \rho \bigg)|_x
	\\[4pt]
	&= \frac\iu\hbar\, \int_{\triangle^2(x;w,v)}\, \mathcal{L}_{\hat{a}} \rho
	\\[4pt]
	&= \frac\iu\hbar\, \bigg( \int_{\triangle^2(x;w,v)}\,
          \iota_{\hat{a}} H + \int_{\partial \triangle^2(x;w,v)}\, \iota_{\hat{a}} \rho \bigg)
	\\[4pt]
	&= \iu\, \big(\chi_{v,w} + \tau_{-w}^*\lambda_v + \lambda_w -
          \lambda_{v+w} \big)|_x(a)\ ,
\end{align}
where we have used Stokes' Theorem and the Cartan formula, as well as
the fact that for 1-forms $A, A' \in \Omega^1(M)$ there is a canonical isomorphism of hermitean line bundles $E_A \otimes E_{A'} \cong E_{A + A'}$ (in the notational conventions of Definition~\ref{def:BGrb_triv}).
Finally we check that with these choices of $\lambda$ and $\Lambda$,
the diagram~\eqref{EQ: Weak coboundary} commutes; we compute
\begin{equation}
\begin{aligned}
	& \exp \bigg(  \frac\iu\hbar\, \int_{\triangle^2(x;v,u)}\, \rho
  + \frac\iu\hbar\, \int_{\triangle^2(x;w,u+v)}\, \rho \\ & \qquad
  \qquad \qquad -\,
  \frac\iu\hbar\, \int_{\triangle^2(x;v+w,u)}\, \rho - \frac\iu\hbar\,
  \int_{\triangle^2(x-u,w,v)}\, \rho - \frac\iu\hbar\, \int_{\triangle^3(x;w,v,u)}\, H \bigg) = 1
\end{aligned}
\end{equation}
for all $x \in M$, where we used~\eqref{Eq: Boundary delta3} and Stokes' Theorem. 
\end{proof}

\subsection{Weak projective 2-representation of magnetic translations}

To construct a weak projective 2-representation of the translation group $\R^d_\rmt$ we define a parallel transport functor on sections of the bundle gerbe $\mathcal{I}_\rho$.
For arbitrary translation vectors $v \in \R^d_\rmt$, connection 1-forms $\eta \in \Omega^1(M, \frh(n))$ and functions $f \in C^\infty(M, \Mat(k \times l))$ for $n,k,l \in \N_0$ it reads as
\begin{equation}
\label{eq:definition of PT of I_rho}
\begin{split}
	\Pa_v \colon \Gamma(M,\mathcal{I}_\rho) &\longrightarrow \Gamma(M,\mathcal{I}_\rho)\ ,
	\\
	\eta &\longmapsto \Pa_v(\eta)
	\qquad \text{with} \quad
	\Pa_v(\eta)|_x \coloneqq \eta|_{x-v} + \frac1\hbar\, \int_0^1\, \rho|_{x - (1-t)\,v}(v, -)\ \dd t \cdot \mathds{1}_n\ ,
	\\
	f & \mapsto \Pa_v(f)
	\qquad \text{with} \quad
	\Pa_v(f)(x) = f(x-v) \ ,
\end{split}
\end{equation}
for all $x\in M$.
We can rewrite the term in $\Pa_v(\eta)$ that contains $\rho$ as
\begin{align}
	\bigg( \int_0^1\, \rho|_{x - (1-t)\,v}(v, -)\ \dd t \cdot \mathds{1}_n \bigg) (a)
	= -\int_{\triangle^1(x;v)}\, \iota_{\hat{a}} \rho \cdot \mathds{1}_n
	\ ,
\end{align}
for all $a\in T_xM$.
We observe that $\Pa_v \colon \Gamma(M, \Ia_\rho) \to \Gamma(M, \Ia_\rho)$ is a twisted $\HVBdl_\triv(M)$-module functor:
for any $\xi \in \Omega^1(M, \frh(k))$, regarded as a connection on a trivial hermitean vector bundle on $M=\R^d$, and $\eta \in \Gamma(M, \Ia_\rho)$ we have
\begin{equation}
	\Pa_v( \xi \otimes \eta) = \tau_{-v}^*\xi \otimes \Pa_v(\eta)
	= \tau(v)[\xi] \otimes \Pa_v(\eta)
\end{equation}
in the notation of Definition~\ref{def:BGrb_triv}, and accordingly for morphisms $f$.

To make this into a weak projective 2-representation we define natural isomorphisms
\begin{equation}
\label{eq:def Pi_{v,w}}
\begin{split}
	\Pi_{v,w} &\colon \Pa_v \circ \Pa_w \Longrightarrow \ell_{\chi_{v,w}} \circ \Pa_{v+w} \ ,
	\\
	\Pi_{v,w|\eta} &\colon  \Pa_v \circ \Pa_w (\eta)
        \longrightarrow \frac1\hbar\, \int_{\triangle^2(-;w,v)}\,  \iota_- H \cdot \mathds{1}_n + \Pa_{v+w}(\eta) \ ,
	\\
	\Pi_{v,w|\eta}(x) &:= \exp\bigg(-\frac\iu\hbar\, 
        \int_{\triangle^2(x;w,v)}\, \rho \bigg) \cdot  \mathds{1}_n
\end{split}
\end{equation}
for all translation vectors $v,w \in \R^d_\rmt$ and points $x \in M$.

\begin{lemma}
For any $v,w \in \R^d_\rmt$, the isomorphism $\Pi_{v,w|\eta}$ is parallel and natural in $\eta$.
\end{lemma}

\begin{proof}
In the notation of \eqref{Condition morphisms} we now have to consider $f = \Pi_{v,w|\eta}$.
We calculate 
\begin{align}
	f^{-1}\, \dd f|_x(a) = \dd \log(\Pi_{v,w|\eta}) |_{x}(a)
	= - \frac\iu\hbar\, \mathcal{L}_{\hat{a}} \bigg( \int_{\triangle^2(-;w,v)}\, \rho\bigg)|_x \cdot  \mathds{1}_n = -\frac\iu\hbar\,  \int_{\triangle^2(x;w,v)}\, \mathcal{L}_{\hat{a}} \rho \cdot  \mathds{1}_n\ .
\end{align}
On the other hand, for the difference $\iu\,(\xi' - \xi)$ in~\eqref{Condition morphisms}, with $\xi$ and $\xi'$ chosen as indicated in~\eqref{eq:def Pi_{v,w}}, we obtain
\begin{align}
	\iu\, (\xi' - \xi)|_x(a) &= \iu\, \bigg(\Pa_{v+w}(\eta) + \frac1\hbar\, \int_{\triangle^2(-;w,v)}\, \iota_- H \cdot \mathds{1}_n - \Pa_v \circ \Pa_w (\eta)\bigg)|_x(a)
	\\[4pt]
	&= \frac\iu\hbar\, \int_{\partial \triangle^2(x;w,v)}\,  \iota_{\hat{a}} \rho \cdot \mathds{1}_n
	+ \frac\iu\hbar\, \int_{\triangle^2(x;w,v)}\, \iota_{\hat{a}} H \cdot \mathds{1}_n
	\\[4pt]
	& = \frac\iu\hbar\, \int_{ \triangle^2(x;w,v)}\, \dd\, \iota_{\hat{a}} \rho \cdot  \mathds{1}_n
	+ \frac\iu\hbar\, \int_{\triangle^2(x;w,v)}\, \iota_{\hat{a}} H \cdot \mathds{1}_n
	\\[4pt]
	&=  \frac\iu\hbar\, \int_{\triangle^2(x;w,v)}\, \mathcal{L}_{\hat{a}} \rho \cdot \mathds{1}_n \ ,
\end{align}
using Stokes' Theorem, the Cartan formula and $H=\dd\rho$.
The naturality follows from the fact that $\Pi_{v,w|\eta}$ is independent of $\eta$ and central in the algebra of matrix-valued functions on $M$.  
\end{proof}

\begin{theorem}
The pair $(\Pa, \Pi)$ defined in~\eqref{eq:definition of PT of I_rho} and~\eqref{eq:def Pi_{v,w}} forms a weak projective 2-representation of the translation group $\R^d_\rmt$ on the $\HVBdl_\triv(M)$-module category $\Gamma(M,\Ia_\rho)$ twisted by the higher weak 2-cocycle $(\chi,\omega,\tau)$ defined in Section~\ref{Sec: Def. of the cocycle}.
\end{theorem}

\begin{proof}
We have to verify that the diagram~\eqref{EQ: Diagram definition weak Projective} commutes, which amounts to commutativity of the diagram
\begin{equation}
\begin{tikzcd}[column sep=2.25cm, row sep=1.25cm]
	\ell_{\chi_{u,v}} \circ \Pa_{u+v} \circ \Pa_w \ar[r,Rightarrow,"\Pi_{u+v,w}"]& \ell_{\chi_{u,v}} \circ \ell_{\chi_{u+v,w}} \circ \Pa_{u+v+w} \ar[dr,Rightarrow,"\ell_{\omega_{u,v,w}}"] &
	\\
	\Pa_u \circ \Pa_v \circ \Pa_w \ar[d,Rightarrow,"{\Pa_u(\Pi_{v,w})}"']  \ar[u,Rightarrow, "\Pi_{u,v}"]  & &  \ell_{\tau(u)[\chi_{v,w}]} \circ \ell_{\chi_{u,v+w}} \circ \Pa_{u+v+w}
	\\
	\Pa_u \circ \ell_{\chi_{v,w}} \circ \Pa_{v+w} \ar[r,Rightarrow,"="'] & \ell_{\tau(u)[\chi_{v,w}]} \circ \Pa_u \circ \Pa_{v+w} \ar[ru,Rightarrow, "\Pi_{u,v+w}", swap] &
\end{tikzcd}
\end{equation}
The functors $\ell_{\chi}$ do not change the functions underlying the isomorphisms $\Pi$, while $\Pa$ acts on these functions only by a translation.
At the level of the underlying functions we thus calculate
\begin{align}
	\Pi_{u+v,w} \circ \Pi_{u,v} \, \circ \, & \, \Pa_u (\Pi_{v,w})^{-1} \circ \Pi_{u,v+w}^{-1}(x)
	\\
	&= \exp \bigg(\frac\iu\hbar\, \int_{\triangle^2(x;v+w,u)}\, \rho + \frac\iu\hbar\, \int_{\triangle^2(x-u;w,v)}\, \rho
	- \frac\iu\hbar\, \int_{\triangle^2(x;u,v)}\, \rho - \frac\iu\hbar\, \int_{\triangle^2(x;u+v,w)}\, \rho \bigg)
	\\[4pt]
	&= \exp \bigg( -\frac\iu\hbar\, \int_{\partial \triangle^3(x;w,v,u)}\, \rho \bigg)
	\\[4pt]
	&=  \exp \bigg( -\frac\iu\hbar\, \int_{ \triangle^3(x;w,v,u)}\, H \bigg)
	\\[4pt]
	&= \omega_{u,v,w}^{-1}(x) \ ,
\end{align}
for all $x\in M$.
Here we have once again made use of the decomposition~\eqref{Eq: Boundary delta3} of the oriented boundary of the 3-simplex.
\end{proof}

\begin{remark}\label{Remark: non-associativity}
There are two different ways to go from the composition of three translation operators to a single translation operator.
Their difference is controlled by $\omega$ as depicted in the diagram
\begin{equation}
\begin{tikzcd}[row sep=1.25cm]
	\Pa_u \circ (\Pa_v \circ \Pa_w) \ar[r,Rightarrow] & \Pa_u \circ (\ell_{\chi_{v,w}} \circ \Pa_{v+w}) = \ell_{\tau_{-u}^* \chi_{u,v}} \circ \Pa_u \circ \Pa_{v+w} \ar[r, Rightarrow] &  \ell_{\tau_{-u}^* \chi_{v,w}} \circ \ell_{\chi_{u,v+w}}  \circ \Pa_{u+v+w} \ar[d,Rightarrow, "\ell_{\omega_{u,v,w}^{-1}}"]
	\\
	(\Pa_u \circ \Pa_v) \circ \Pa_w \ar[r,Rightarrow] & \ell_{\chi_{u,v}} \circ \Pa_{u+v} \circ \Pa_w \ar[r, Rightarrow] & \ell_{\chi_{u,v}} \circ \ell_{\chi_{u+v,w}} \circ \Pa_{u+v+w}
\end{tikzcd}
\end{equation}
This is the implementation of nonassociativity in the higher categorical framework.  
\end{remark}

\begin{remark}
Using the expressions in~\eqref{eq:pushforward of wp2reps} to push forward the weak projective 2-representation along the higher weak 2-coboundary defined in Proposition~\ref{Prop: Trivialisation of the cocycle} yields the honest 2-representation of $\R^d_\rmt$ on $\Gamma(M,\Ia_\rho)$ by bare pullbacks, which is completely associative.
This should be understood as a higher categorical analogue of switching between kinematical and canonical momentum operators as discussed in Section~\ref{Sec: Case without sources}. 
\end{remark}

\begin{remark}
The weak projective 2-representation $(\Pa, \Pi)$ together with the higher weak 2-cocycle $(\chi, \omega, \tau)$ give an independent derivation of the nonassociativity of magnetic translations in generic backgrounds of magnetic charge, as calculated originally by~\cite{Jackiw:1984rd}.
In this latter approach the $C^\infty(\R^3,\sfU(1))$-valued 3-cocycle $\omega$ on $\R_\rmt^3$ that appears as part of the higher weak 2-cocycle $(\chi, \omega, \tau)$ was derived by purely algebraic means from the commutation and association relations~\eqref{eq:can comm with B-field} and~\eqref{eq:3-can associator with B-field}, with no description of the quantities that nonassociative magnetic translations act on.
Here, in contrast, we have arrived at the 3-cocycle by first answering that question -- it is the category of sections of a bundle gerbe $\Ia_\rho$ -- and then representing magnetic translations by means of parallel transport on $\Ia_\rho$.
Thus the parallel transport $\Pa$ on the bundle gerbe $\Ia_\rho$ relates to the 3-cocycle $\omega$ in complete analogy to how the parallel transport $P$ on the line bundle $L$ in Section~\ref{Sec: Case without sources} is related to the 2-cocycle derived in e.g.~\cite{Hannabuss:2017ion,Soloviev:Dirac_Monopole_and_Kontsevich}.
This suggests an interpretation of sections of bundle gerbes as a generalised model for the quantum state space of a charged particle in generic distributions of magnetic charge.
\end{remark}

\subsection{Examples}

Let us now look at two particular examples.
The first example is a consistency check to some extent -- we investigate the case where the magnetic charge distribution $H=\dd\rho$ vanishes identically.
In the second example we consider a constant distribution $H=\tilde H$ of magnetic charge.
This latter example cannot be treated in the gauge theory formalism of classical electromagnetism;
while a finite collection of Dirac monopoles may be treated by removing their locations to avoid singularities, there is no line bundle that can realise a non-vanishing smooth distribution of magnetic charge.

\subsubsection*{No magnetic charge}

Let us first assume that $H = \dd \rho = 0$.
Then we can find a suitable 1-form $A \in \Omega^1(M)$ such that $\rho = \dd A$.
The condition $H = 0$ implies that the higher weak 2-cocycle $(\chi,\omega, \tau)$ is the trivial higher weak 2-cocycle twisted by $\tau$.
That is, $(\chi,\omega, \tau) = (\chi^0,\omega^0, \tau)$, where the action $\tau$ of the translation group on $\mathcal{R} = \HVBdl_\triv(M)$ is still non-trivial. 
The nonassociativity discussed in Remark~\ref{Remark: non-associativity} is absent, however, since $\omega$ is trivial.
Nevertheless, there still remains the non-trivial action of the translation group and a non-trivial 2-cocycle described by $\Pi$, which is part of the weak projective 2-representation $(\Pa, \Pi)$.
This is precisely the 2-cocycle introduced in~\eqref{EQ: Definition 2 cocycle} and derived in e.g.~\cite{Hannabuss:2017ion,Soloviev:Dirac_Monopole_and_Kontsevich}, which governs the weak projective representation in the classical case. Then~\eqref{EQ: Diagram definition weak Projective} reduces to the 2-cocycle condition.
We thus obtain

\begin{proposition}
If $\Ia_\rho$ is a flat bundle gerbe on $M$ and $A \in \Omega^1(M)$ satisfies $\dd A = \rho$, then the data of the weak projective 2-representation $(\Pa, \Pi)$ reduce to yield a 2-cocycle $\Pi \in \sfC^2(\R^d_\rmt, C^\infty(M, \sfU(1)))$.
This 2-cocycle agrees with the 2-cocycle~\eqref{EQ: Definition 2 cocycle} that the parallel transport on line bundles produces when starting from the hermitean line bundle $E_A$ with connection on $M$.
\end{proposition}

Similarly to the discussion in Section~\ref{Sec: Case without sources} we can trivialise this weak projective 2-representation.

\begin{proposition}
Let $H = \dd \rho = 0$, and let $A \in \Omega^1(M)$ satisfy $\dd A = \rho$.
Denote by $(\Pa^0 , \Pi^{0} = \id)$ the trivial weak projective 2-representation of $\R^d_\rmt$ on $\Gamma(M, \Ia_{\dd A})$ with respect to $\tau$.
There is an isomorphism of weak projective 2-representations $(\ell_A, \eta) \colon (\Pa,\Pi)\longrightarrow (\Pa^0,\Pi^0)$, with
\begin{align*}
	\ell_A \colon \Gamma(M,\mathcal{I}_{\dd A}) &\longrightarrow \Gamma(M,\mathcal{I}_{\dd A})\ ,
	\\
	\eta &\longmapsto \eta + A \cdot \mathds{1}_n\ ,
	\\
	f &\longmapsto f \ ,
\end{align*} 
and 
\begin{align*}
	m_v &\colon \ell_A \circ \Pa_v \Longrightarrow \Pa^0_v \circ \ell_A \ ,
	\\
	m_{v|\eta} &\colon \tau_{-v}^* \eta + A \cdot \mathds{1}_n + \frac1\hbar\, \int_{\triangle^1(-;v)}\, \rho \cdot \mathds{1}_n
	\longrightarrow \tau_{-v}^*(\eta + A \cdot \mathds{1}_n)\ ,
	\\
	m_{v| \eta}(x) &\coloneqq \exp \bigg( \frac\iu\hbar \, \int_{\triangle^1(x;v)}\, A \bigg) \cdot \One_n
\end{align*}
for all $\eta \in \Omega^1(M, \frh(n))$, $v \in \R^d_\rmt$, and $x \in M$.
\end{proposition}

\begin{proof}
Let $a \in T_x M$ be a tangent vector at $x \in M$.
We start by checking that $m_v$ is parallel:
\begin{align}
	\dd \bigg( \frac\iu\hbar \, \int_{\triangle^1(-;v)}\, A \bigg)|_x (a)
	&= \frac\iu\hbar\, \mathcal{L}_{\hat{a}} \bigg(\int_{\triangle^1(-;v)}\, A\bigg)|_x
	\\[4pt]
	&= \frac\iu\hbar\, \int_{\triangle^1(x;v)}\, \mathcal{L}_{\hat{a}} A
	\\[4pt]
	&= \frac\iu\hbar\, \bigg( \int_{\triangle^1(x;v)}\, \iota_{\hat{a}} \rho + A|_x(a)-A|_{x-v}(a) \bigg)\ .
\end{align}
To show the commutativity of the diagrams in Definition \ref{Morphisms of weak projective representations} in this case we have to check the identity $m_w \circ m_v= m_{v+w} \circ \Pi_{v,w}$ for all $v,w \in \R^d_\rmt$.
Inserting the definitions we see that this equation reduces to~\eqref{Eq: quasi-projective representation}.
\end{proof}

\subsubsection*{Constant magnetic charge}

Let us now consider a non-vanishing but constant magnetic charge
$H=\tilde H$.  Then the 3-cocycle $\omega$ factors through $\sfU(1)$ regarded as a trivial $\R^d_\rmt$-module, i.e. there exists a 3-cocycle $\tilde{\omega} \colon \R^d_\rmt \times \R^d_\rmt \times \R^d_\rmt\longrightarrow \sfU(1)$ that makes the diagram
\begin{equation}
\begin{tikzcd}
	\R^d_\rmt \times \R^d_\rmt \times \R^d_\rmt \ar[r, "\omega"] \ar[d, dashed, "\tilde{\omega}",swap]  & C^\infty(M, \sfU(1)) \\
	  \sfU(1) \ar[ru,hookrightarrow, "\imath"'] 
\end{tikzcd} 
\end{equation} 
commute. In this case the higher weak 2-cocycle $(\tilde\chi, \tilde\omega)$ is given by
\begin{equation}\label{eq:tildeomega}
	\tilde\chi_{v,w}|_x(a) =  \frac1{2\hbar} \, \tilde H(a,v, w) \qquad
        \mbox{and} \qquad \tilde\omega_{u,v,w}(x) = \exp \bigg(
        \frac\iu{6\hbar}\, \tilde H(u ,v, w) \bigg)
\end{equation}
for any translation vectors $u,v,w \in \R^d_\rmt$, points $x \in M$, and tangent vectors $a \in T_xM$.
The weak projective 2-representation for the bundle gerbe $\Ia_{\tilde\rho}$
with the magnetic field
\begin{equation}\label{eq:rhoHconstant}
	\tilde\rho_{ij}(x) = \frac13\, \sum_{k=1}^d\, \tilde H_{ijk}\, x^k
\end{equation}
takes the form
\begin{align}
	\tilde\Pa_v \colon \Gamma(M,\mathcal{I}_\rho)& \longrightarrow \Gamma(M,\mathcal{I}_\rho)\ ,
	\\
	\eta|_{x} & \longmapsto (\tau_{-v}^*\eta)|_{x} +
                    \frac{1}{3\hbar}\, \tilde H(v ,-, x) \ ,
	\\
	f & \longmapsto \tau_{-v}^*f  
\end{align}
with coherence isomorphisms
\begin{align}\label{eq:tildePi}
	\tilde\Pi_{v,w}(x) = \exp \bigg(- \frac\iu{6\hbar} \, \tilde H(x ,v, w)
  \bigg) \ ,
\end{align}
for all $\eta\in \Omega^1(M)$.
Thus the weak projective 2-representation induced by the parallel
transport $\Pa$ on the bundle gerbe $\Ia_{\tilde\rho}$ is non-trivial even for
the simple magnetic field \eqref{eq:rhoHconstant}.
We emphasise that despite its very simple form this magnetic
field cannot be treated using line bundles -- in order to capture the
non-trivial 3-form field strength $\tilde H$ one has to employ the
formalism of bundle gerbes if one wishes to describe quantum states of
the charged particle geometrically.

\subsection{Deformation quantisation}
\label{sec:defquant}

Despite the complexity and somewhat abstract setting of the geometric formalism above, the bivector field
\eqref{eq:varthetarho} and corresponding twisted Poisson brackets
\eqref{eq:twistedbrackets} can still be treated concretely through (formal)
deformation quantisation on the space of smooth functions
$C^\infty(\frM,\C)$ on phase space $\frM$ for arbitrary smooth
distributions of magnetic charge
$H=\dd\rho\in\Omega^3(M)$. Generically, in this case the noncommutative and
nonassociative star product $\star_H$ is a product on the algebra of
formal power series $C^\infty(\frM,\C)[[\hbar]]$ defined for  two
smooth functions $f,g$ on $\frM$ by
\begin{equation}
f\star_H g = f\, g + \frac{\iu\,\hbar}2\, \{f,g\}_\rho +
\sum_{n\geq2}\, \frac{(\iu\,\hbar)^n}{n!} \ {\rm b}_n(f,g) \ ,
\end{equation}
where the coefficients ${\rm b}_n$ are bidifferential operators. This
was first constructed by~\cite{MSS:NonGeo_Fluxes_and_Hopf_twist_Def}
using the Kontsevich formalism, which provides an explicit
construction of the bidifferential operators ${\rm b}_n$ in terms of
integrals on configuration spaces of the hyperbolic
plane.\footnote{See~\cite{Soloviev:Dirac_Monopole_and_Kontsevich} for a
  treatment of the Dirac monopole field in this setting.} The
Kontsevich formality construction
also quantises the trivector field
\eqref{eq:Schoutenbracket} and corresponding Jacobiators
\eqref{eq:Jacobiators} to the $3$-bracket measuring nonassociativity of three smooth functions
$f,g,h$ given by
\begin{equation}
[f,g,h]_{\star_H} = -\hbar^2\, \{f,g,h\}_\rho + \sum_{n\geq3}\,
\frac{(\iu\,\hbar)^n}{n!} \ {\rm
  t}_n(f,g,h) \ ,
\end{equation}
where ${\rm t}_n$ are tridifferential operators.

These formal power series expansions simplify drastically in the case
of a constant magnetic charge distribution $H=\tilde H$, making the
derivation of explicit expressions
possible~\cite{MSS:NonGeo_Fluxes_and_Hopf_twist_Def}. In particular,
in~\cite{Mylonas:2013jha} it is shown that there is a strict
deformation quantisation formula for the nonassociative star product which is
formally identical to the twisted convolution integral
\eqref{eq:MoyalWeylconstant} for the Moyal-Weyl star product:
\begin{equation}
f\star_{\tilde H} g = \frac1{(\pi\,\hbar)^{2d}}\, \int_\frM \ \int_\frM\
\e^{-\frac{2\,\iu}\hbar\, \sigma_{\tilde\rho}(Y,Z)}\, f(X-Y)\, g(X-Z)\ \dd Y\
\dd Z \ ,
\end{equation}
where here $\sigma_{\tilde\rho}$ is the almost symplectic form \eqref{eq:sigmarho}
corresponding to the magnetic field \eqref{eq:rhoHconstant}. In this
case, the functions
\begin{equation}
\Wscr(\Pa_v):=\exp\Big(\frac\iu\hbar \, \langle
  p,v\rangle \Big)
\end{equation}
explicitly realise the algebraic relations of the weak
projective 2-representation of nonassociative magnetic translation
operators through~\cite{Mylonas:2013jha,Szabo:Magnetic_monopoles_and_NAG}
\begin{equation}
\Wscr(\Pa_v)\star_{\tilde H}\Wscr(\Pa_w) = \tilde\Pi_{v,w}(x) \
\Wscr(\Pa_{v+w}) 
\end{equation}
and
\begin{equation}
\big(\Wscr(\Pa_u)\star_{\tilde H}\Wscr(\Pa_v)\big) \star_{\tilde H} \Wscr(\Pa_w) =
\tilde\omega_{u,v,w}(x) \ \Wscr(\Pa_u)\star_{\tilde H} \big(
\Wscr(\Pa_v)\star_{\tilde H}\Wscr(\Pa_w) \big) \ ,
\end{equation}
where $\tilde\omega$ and $\tilde\Pi$ are given in
\eqref{eq:tildeomega} and \eqref{eq:tildePi}, respectively. While this
evidently seems to suggest a higher version of the Weyl correspondence
discussed in Section~\ref{Sec: Case without sources}, it is not clear
to us at this stage how to make this precise: what is missing is a
suitable higher version of a magnetic Weyl system
\eqref{eq:Weylsystem} that would lead to a
quantisation map $f\mapsto \mathcal{O}_f$ taking phase space functions to
suitable functors defined by parallel transport on sections of the
bundle gerbe $\Ia_{\tilde\rho}$. A categorification of the magnetic Weyl correspondence is also discussed
in~\cite[Section~4]{MSS:NonGeo_Fluxes_and_Hopf_twist_Def} by
integrating the $L_\infty$-algebra of
the twisted magnetic Poisson structure to a suitable Lie 2-group into which 
$C^\infty(\frM,\C)$ embeds as an algebra object.

\section{Covariant differentiation in bundle gerbes}
\label{sect:Covariant derivatives of sections}

Let $\Ia_\rho$ be a topologically trivial bundle gerbe on $M=\R^d$ corresponding to a magnetic field $\rho\in\Omega^2(M)$.
In Section~\ref{sect:weak projective rep of translations} we have seen how the translation group $\R^d_\rmt$ acts on the category $\Gamma(M, \Ia_\rho)$ of smooth global sections of $\Ia_\rho$, and thereby on the 2-Hilbert space $(\Gamma(M,\Ia_\rho), \<-,-\>)$. With an eye to understanding better what a higher version of the magnetic Weyl correspondence discussed in Section~\ref{sec:defquant} might involve, we can examine infinitesimal translations, or derivatives, which correspond to momentum operators in the applications to quantum mechanics. 
In this section we analyse what it means for a section of a bundle gerbe $\Ia_\rho$ on $M$ to be covariantly constant and carry out first steps towards understanding momentum operators in this higher geometric context.

\subsection{Homotopy fixed points}

Throughout this section we consider a general hermitean vector bundle $(E, \nabla^E)$ with connection on $M=\R^d$.
As pointed out in Remarks~\ref{rmk:BGrb_triv makes sense on every M} and~\ref{eq:sections are equiv to Hilb}, we may view $(E,\nabla^E)$ as a section of $\Ia_\rho$.
We would like to find a notion of when a section $(E,\nabla^E) \in \Gamma(M, \Ia_\rho)$ is parallel, in order to then understand the covariant derivative of a general section as an obstruction to it being parallel.

We start again by considering sections of line bundles.
Let $(L,\nabla^L)$ be a hermitean line bundle with connection on $M=\R^d$.
The translation group $\R^d_\rmt$ acts on the space of sections $\Gamma(M, L)$ by
\begin{equation}
	(P_v \psi)|_{x} \coloneqq P^{\nabla^L}_{\triangle^1(x;v)} (\psi|_{x-v})\ ,
\end{equation}
where $P^{\nabla^L}$ is the parallel transport on $L$ that corresponds to $\nabla^L$ and $\psi \in \Gamma(M,L)$ is a smooth section of $L$.
Let $v \in \R^d_\rmt$ be an arbitrary translation vector with associated global vector field $\hat{v} \in \Gamma(M, TM)$, and let $\<v\> \subset \R^d_\rmt$ denote the subgroup generated by $v$.
That is, $\<v\>=\{s\,v\,|\,s\in\R\}$ is the group of translations on $M$ in the direction of $v$.
A section $\psi \in \Gamma(M, L)$ is covariantly constant in the
direction of $v$, i.e. $\nabla^L_{\hat{v}} \psi = 0$, if and only if
$\psi$ is invariant under, or a \emph{fixed point} for, the
restriction of the action of the magnetic translations $P_{(-)}$ to the subgroup $\<v\> \subset \R_\rmt^d$.

In the categorified setting there is an appropriately weakened notion of invariance under a group action.

\begin{definition}
\label{def:hofpt}
Let $(\Theta, \Pi)$ be an action of a group $\sfG$ on a category $\Ca$ as in Definition \ref{Def: Group action on categories}.
A \emph{homotopy fixed point for $\Theta$} is an object $C \in \Ca$ together with a collection of isomorphisms $\epsilon_g \colon \Theta_g (C) \to C$ for all $g \in \sfG$ satisfying the coherence condition
\begin{equation}
\begin{tikzcd}[row sep=1.25cm, column sep = 1.75cm]
	\Theta_h \circ \Theta_g (C) \ar[r, "\Theta_h(\epsilon_g)"] \ar[d, "\Pi_{h,g}"'] & \Theta_h( C) \ar[d, "\epsilon_h"]
	\\
	\Theta_{h\,g}( C) \ar[r, "\epsilon_{h\,g}"'] & C
\end{tikzcd}
\end{equation}
for all $g,h\in\sfG$.
\end{definition}

We therefore investigate when a section $(E,\nabla^E) \in \Gamma(M, \Ia_\rho)$  can be endowed with a homotopy fixed point structure for the action of $\<v\> \subset \R^d_\rmt$ on $\Gamma(M, \Ia_\rho)$ that is given by the parallel transport $\Pa$, defined in~\eqref{eq:definition of PT of I_rho}, of the bundle gerbe $\Ia_\rho$.
Recalling that $\tau_v \colon M \to M$, $x \mapsto x+v$ denotes the translation by $v$, we set
\begin{equation}
	\Pa_v (E, \nabla^E) \coloneqq \tau_{-v}^* (E, \nabla^E)
        \otimes \bigg( M \times \C \ , \ \dd - \frac\iu\hbar\, \int_{\triangle^1(-;v)}\, \rho \bigg)\ ,
\end{equation}
which is the generalisation of~\eqref{eq:definition of PT of I_rho} to possibly non-trivial hermitean vector bundles.
Defining $\Pa_v \psi \coloneqq \tau^*_{-v} \psi$ for a morphism $\psi \colon E \to F$ of vector bundles we turn $\Pa_v$ into a functor $\HVBdl(M) \to \HVBdl(M)$.
The functor $\Pa_v$ preserves unitary and parallel morphisms.
The parallel transport on $E$ induces isomorphisms of vector bundles
\begin{equation}
\label{eq:hofpt morphism}
	P^{\nabla^E}_v \colon \Pa_v (E, \nabla^E) \to (E, \nabla^E)\ ,
	\quad
	E|_{x - v} \ni e \mapsto P^{\nabla^E}_{\triangle^1(x;v)} (e) \in E|_{x}\ .
\end{equation}
This structure is even coherent when carrying out several translations in the same direction:
let $s, t \in \R$ and recall the notation $E_\eta \coloneqq (M \times \C^n, \dd + \iu\, \eta) \in \HVBdl_\triv(M)$ from Definition~\ref{def:BGrb_triv}, where $\eta \in \Omega^1(M, \frh(n))$.
Consider the diagram that arises from the definition of $\Pi$
in~\eqref{eq:def Pi_{v,w}} given by
\begin{equation}
\label{eq:hofpt diagram for sv tv}
\small
\begin{tikzcd}[row sep=1.25cm, column sep = 2.0cm]
	E_{\chi_{s\,v,t\,v}} \otimes \Pa_{s\,v}\circ \Pa_{t\,v} (E, \nabla^E) \ar[r, "\Pa_{s\,v} (P^{\nabla^E}_{t\,v})"] \ar[d, "\Pi_{s\,v,t\,v}"'] & \Pa_{s\,v} (E, \nabla^E) \ar[d, "P^{\nabla^E}_{s\,v}"]
	\\
	\Pa_{(s+t)\,v} (E, \nabla^E) \ar[r, "P^{\nabla^E}_{(s+t)\,v}"'] & E
\end{tikzcd}
\quad = \quad
\begin{tikzcd}[row sep=1.25cm, column sep = 2.0cm]
	\Pa_{s\,v}\circ \Pa_{t\,v} (E, \nabla^E) \ar[r, "\Pa_{s\,v} (P^{\nabla^E}_{t\,v})"] \ar[d, "\id"'] & \Pa_{s\,v} (E, \nabla^E) \ar[d, "P^{\nabla^E}_{s\,v}"]
	\\
	\Pa_{(s+t)\,v} (E, \nabla^E) \ar[r, "P^{\nabla^E}_{(s+t)\,v}"'] & E
\end{tikzcd}
\normalsize
\end{equation}
Here we have used the fact that, because the two translation vectors
$s\,v$ and $t\,v$ are parallel, the 2-simplices
$\triangle^2(-;s\,v,t\,v)$ are all degenerate, thus making $\chi_{s\,v,t\,v}$ as well as $\Pi_{s\,v,t\,v}$ trivial.
The diagram on the right-hand side then commutes due to the fact that parallel transports in vector bundles are compatible with concatenation of paths.
This result can be summarised as 

\begin{proposition}
\label{st:directional hofpt structure on sections}
For any $v \in \R^d_\rmt$, the morphisms $P^{\nabla^E}_{s\,v}$ for
$s\in\R$ provide a homotopy fixed point structure on $(E,\nabla^E)$
for the action of the subgroup $\<v\> \subset \R^d_\rmt$ on the
category of sections $\Gamma(M, \Ia_\rho)$ of the bundle gerbe $\Ia_\rho$.
\end{proposition}

\begin{remark}
\label{rmk:par sections and hol}
The diagram~\eqref{eq:hofpt diagram for sv tv} does not commute for
arbitrary translations $v,w \in \R^d_\rmt$ if the parallel transport
$\Pa$ of $\Ia_\rho$ has non-trivial holonomy line bundle $E_{\chi_{v,w}}$, in analogy to the obstruction on the existence of parallel sections posed by the holonomy of a connection on a line bundle.
This will be made precise in Theorem~\ref{st:existence of parallel
  sections} below.
\end{remark}

\subsection{Parallel homotopy fixed points and fake curvature}

We would like to understand a homotopy fixed point structure for $\<v\>$ on a section $(E,\nabla^E) \in \Gamma(M, \Ia_\rho)$ as a notion of $(E,\nabla^E)$ being covariantly constant in the direction defined by $v$.
However, Proposition~\ref{st:directional hofpt structure on sections} states that there exists a homotopy fixed point structure on $(E,\nabla^E)$ for any translation vector $v \in \R^d_\rmt$, so that every section of $\Gamma(M,\Ia_\rho)$ would be parallel.
In Remark~\ref{rmk:par sections and hol}, in contrast, we observed
that the holonomy line bundle $E_{\chi_{v,w}}$ poses an obstruction to
the existence of global homotopy fixed points.
Therefore, the homotopy fixed point structures from Proposition~\ref{st:directional hofpt structure on sections} cannot be the correct notion of covariant constancy for sections of $\Ia_\rho$ yet.

The resolution of this contradiction is that, while the morphism
$P_v^{\nabla^E}$ defined in~\eqref{eq:hofpt morphism} is always a unitary isomorphism, it is not necessarily parallel.
Following notions of gauge theory, we regard two sections in the 2-Hilbert space $\Gamma(M, \Ia_\rho)$ given by hermitean vector bundles $(E,\nabla^E)$ and $(F, \nabla^F)$ with connection as equivalent if they differ only by a gauge transformation, i.e. by a unitary \emph{parallel} isomorphism $\psi \colon E \to F$.
The obstruction to $P^{\nabla^E}_v$ being parallel can be computed as follows:
let $w \in \R^d_\rmt$ be an arbitrary translation vector.
Then
\begin{equation}
\begin{aligned}
\big( P^{\nabla^E}_v \big)^{-1} \circ \big( P^{\nabla^E}_w \big)^{-1} \circ P^{\nabla^E}_v \circ P^{\Pa_v(E,\nabla^E)}_w\big|_x
	= \hol \big( (E, \nabla^E), \partial\, \square^2(x;v,w) \big)
        \cdot \exp \bigg( -\frac\iu\hbar\, \int_{\square^2(x;v,w)} \, \rho \bigg)\ ,
\end{aligned}
\end{equation}
where $\square^2(x; v, w) \subset \R^d$ is the parallelogram in $\R^d$ with corners $x - (v+w), x-v, x-w$ and $x$.
Thus, introducing parameters by replacing $v $ with $t\, v$ and $w $
with $s\, w$ for $t, s \in (-1,1)$, and taking the limit $s,t \to 0$,
we can derive the obstruction to $P^{\nabla^E}_v$ being parallel.

\begin{definition}
Let $\Ia_\rho \in \BGrb^\nabla_\triv(M)$ and $(E,\nabla^E) \in \HVBdl(M)$.
The \emph{(higher) covariant derivative of $(E,\nabla^E)$ in the direction $v \in \R^d_\rmt$} is the $\End(E)$-valued 1-form
\begin{equation}
	\mbf\nabla^\rho_{\hat{v}} (E, \nabla^E) \coloneqq \iota_{\hat{v}} \Big( F_{\nabla^E} - \frac1\hbar\, \rho \cdot \One \Big) \ ,
\end{equation}
where $F_{\nabla^E}$ is the curvature of $\nabla^E$.
The $\End(E)$-valued 2-form $F_{\nabla^E} - \frac1\hbar\, \rho \cdot \One$ is called the \emph{fake curvature of $(E,\nabla^E)$} when $(E, \nabla^E)$ is regarded as a section of $\Ia_\rho$.
If the section $(E, \nabla^E)$ is parallel, i.e. if it satisfies $F_{\nabla^E} - \frac1\hbar\, \rho \cdot \One = 0$, we equivalently say that it satisfies the \emph{fake curvature condition}.
\end{definition}

We call the homotopy fixed point structure $P^{\nabla^E}$ on $(E,\nabla^E)$ for the action of $\<v\>\subset\R^d_\rmt$ a \emph{parallel homotopy fixed point structure} if $P^{\nabla^E}_{s\,v}$ is a parallel morphism of vector bundles for all $s \in \R$.
We have thus proved

\begin{theorem}
\label{st:cov der and par hofpt structures}
Let $v \in \R^d_\rmt$ be a translation vector and let $(E, \nabla^E) \in \Gamma(M, \Ia_\rho)$ be a section of $\Ia_\rho$.
Then \smash{$P^{\nabla^E}$} is a parallel homotopy fixed point structure on $(E,\nabla^E)$, for the action of the group $\<v\>$ of translations in the direction of $v$ via the parallel transport $\Pa$ of $\Ia_\rho$, if and only if $\mbf\nabla^\rho_{\hat{v}}(E,\nabla^E) = 0$, i.e.~precisely if $(E,\nabla^E)$ is covariantly constant in the direction $v$.
\end{theorem}

This provides a novel approach to the fake curvature condition, which deepens the understanding of bundle gerbes with connections as higher line bundles with connections.
We can define a $T^*M \otimes \End(E)$-valued 1-form
\begin{equation}
	\bfdd_\rho (E, \nabla^E) \coloneqq F_{\nabla^E} - \frac1\hbar\, \rho \cdot \One = \sum_{i=1}^d\, \hat e{}^i \otimes \mbf\nabla^\rho_{\hat{e}_i} (E, \nabla^E)\ ,
\end{equation}
where as before $(e_i)_{i = 1, \ldots, d}$ is the standard basis of $\R^d$ and $\hat e^i$ is the dual 1-form of the vector field $\hat{e}_i$.
The expression for the covariant exterior differential $\bfdd_{\rho=0}(E, \nabla^E)$ of a higher function $(E, \nabla^E)$ now perfectly parallels the expression for the de~Rham differential of an ordinary function.

Moreover, we can now properly understand the curvature $H = \dd \rho$ of the bundle gerbe $\Ia_\rho$, and thus by~\eqref{eq:def Pi_{v,w}} its holonomy line bundle $E_{\chi_{v,w}}$, as an obstruction to the existence of parallel sections.

\begin{theorem}
\label{st:existence of parallel sections}
Let $\rho \in \Omega^2(M)$ be a magnetic field on $M=\R^d$.
The bundle gerbe $\Ia_\rho$ admits a parallel section if and only if it is flat, i.e. precisely if $H = \dd \rho = 0$.
\end{theorem}

\begin{proof}
If $H = \dd \rho = 0$, then there exists a 1-form $A \in \Omega^1(M)$ such that $\dd A = \rho$.
Set
\begin{equation}
	(E, \nabla^E) = E_A = \Big(M \times \C \ , \ \dd + \frac\iu\hbar \, A\Big)\ .
\end{equation}
Then $\bfdd_\rho E_A = 0$, i.e. $E_A \in \Gamma(M, \Ia_\rho)$ is parallel.

Conversely, let $(E,\nabla^E)$ be a section of $\Ia_\rho$ with $\bfdd_\rho(E,\nabla^E) = 0$.
This is equivalent to $(E,\nabla^E)$ being fake flat, i.e. to $F_{\nabla^E} = \frac1\hbar\,\rho\cdot\One$.
If $\det(E)$ denotes the determinant line bundle of $E$ with connection $\nabla^{\det(E)}$ induced by $\nabla^E$, then $F_{\nabla^{\det(E)}} = \frac{\rk(E)}\hbar \, \rho$, where $\rk(E)$ is the rank of $E$.
We compute
\begin{equation}
	\frac{\rk(E)}\hbar \, \dd \rho = \dd F_{\nabla^{\det(E)}} = 0\ ,
\end{equation}
and the result follows.
\end{proof}

Recall that connections on a hermitean vector bundle $E\to M$ form an affine
space over the vector space $\Omega^1(M, \End_\frh(E))$, where
$\End_\frh(E)$ is the bundle of hermitean endomorphisms of $E$.

\begin{definition}
A \emph{tangent vector of $\HVBdl(M)$ at $(E,\nabla^E)$} is a triple $(E, \nabla^E, \nu)$, where $\nu \in \Omega^1(M, \End_\frh(E))$.
A \emph{morphism of tangent vectors} $(E, \nabla^E, \nu) \to (E', \nabla^{E'}, \nu')$ is a pair $(\psi, \psi^{(1)})$ of morphisms of vector bundles $\psi, \psi^{(1)} \colon E \to E'$.
A pair of morphisms $(\psi, \psi^{(1)})$ is \emph{parallel} if%
\begin{equation}
	\nabla^{\Hom(E,E')} \psi^{(1)} = \psi\, \nu - \nu'\, \psi\ .
\end{equation}
The direct sum and tensor product of two tangent vectors are given by
\begin{equation}
\begin{split}
	(E', \nabla^{E'}, \nu') \oplus (E, \nabla^E, \nu) &= (E' \oplus E, \nabla^{E'} \oplus \nabla^E, \nu' \oplus \nu)\ ,
	\\[4pt]
	(E', \nabla^{E'}, \eta) \otimes (E, \nabla^E, \nu) &= (E' \otimes
        E, \nabla^{E'} \otimes \One + \One \otimes \nabla^E, \nu' \otimes
        \One + \One \otimes \nu)\ .
\end{split}
\end{equation}
On morphisms, these operations read as
\begin{align}
	(\psi, \psi^{(1)}) \oplus (\phi, \phi^{(1)}) &= \big( \psi \oplus \psi,\, \psi^{(1)} \oplus \phi^{(1)} \big)\ ,
	\\[4pt]
	(\psi, \psi^{(1)}) \otimes (\phi, \phi^{(1)}) &= \big( \psi \otimes \One + \One \otimes \phi,\, \psi^{(1)} \otimes \One + \One \otimes \phi^{(1)} \big)\ .
\end{align}
This defines the \emph{rig category of tangent vectors} $T (\HVBdl(M))$ to $\HVBdl(M)$.
\end{definition}

The condition of a morphism $(\psi, \psi^{(1)})$ being parallel is equivalent to
\begin{equation}
	(\nabla^{E'} + t \, \nu') \circ (\psi + t \, \psi^{(1)}) -
        (\psi + t \, \psi^{(1)}) \circ (\nabla^E + t \, \nu) = O(t^2)\ ,
\end{equation}
for $t\in\R$.
If $\psi \colon (E, \nabla^E) \to (E', \nabla^{E'})$ is parallel, then so is
\begin{equation}
	\bfdd_\rho \psi \coloneqq (\psi, \psi^{(1)} = \psi)\ .
\end{equation}
This turns the covariant derivative $\mbf\nabla^\rho_{\hat{v}}$ into a functor
\begin{equation}
	\mbf\nabla^\rho_{\hat{v}} \colon \HVBdl(M) \to T \big( \HVBdl(M) \big)\ .
\end{equation}

We can tensor a tangent vector by a vector bundle using the zero
section to get
\begin{equation}
	(E', \nabla^{E'}) \otimes (E, \nabla^E, \nu) \coloneqq (E', \nabla^{E'}, 0) \otimes (E, \nabla^E, \nu)
	= (E' \otimes E, \nabla^{E' \otimes E}, \One \otimes \nu)\ ,
\end{equation}
where we have abbreviated the tensor product connection by $\nabla^{E' \otimes E}$.
Now if $(E_0, \nabla^{E_0}) \in \Gamma(M, \Ia_0)$ and $(E, \nabla^E) \in \Gamma(M, \Ia_\rho)$, the covariant derivative $\bfdd_\rho$ satisfies the following Leibniz rule:
on the one hand, we have
\begin{equation}
\begin{split}
	\bfdd_\rho \big( (E_0, \nabla^{E_0}) \otimes (E, \nabla^E) \big) &= \bfdd_\rho (E_0 \otimes E, \nabla^{E_0 \otimes E})
	\\[4pt]
	&=\bigg(E_0 \otimes E, \nabla^{E_0 \otimes E}, F_{\nabla^{E_0 \otimes E}} - \frac\iu\hbar\, \rho \cdot \One\bigg)
	\\[4pt]
	&= \bigg( E_0 \otimes E, \nabla^{E_0 \otimes E}, F_{\nabla^{E_0}} \otimes \One + \One \otimes \Big(F_{\nabla^E} - \frac\iu\hbar\, \rho \cdot \One\Big) \bigg)\ ,
\end{split}
\end{equation}
while at the same time
\begin{equation}
\begin{split}
	&\big( \bfdd_0 (E_0, \nabla^{E_0}) \big) \otimes (E, \nabla^E) + (E_0, \nabla^{E_0}) \otimes \bfdd_\rho (E, \nabla^E)
	\\
	& \qquad \qquad \qquad = (E_0, \nabla^{E_0}, F_{\nabla^{E_0}}) \otimes (E, \nabla^E, 0) + (E_0, \nabla^{E_0}, 0) \otimes \bigg(E, \nabla^E, F_{\nabla^E} - \frac\iu\hbar\, \rho \cdot \One\bigg)
	\\[4pt]
	& \qquad \qquad \qquad = (E_0 \otimes E, \nabla^{E_0 \otimes E}, F_{\nabla^{E_0}} \otimes \One) + \bigg( E_0 \otimes E, \nabla^{E_0 \otimes E}, \One \otimes \Big(F_{\nabla^E} - \frac\iu\hbar\, \rho \cdot \One\Big) \bigg)
	\\[4pt]
	& \qquad \qquad \qquad = \bigg( E_0\otimes E, \nabla^{E_0 \otimes E}, F_{\nabla^{E_0}} \otimes \One + \One \otimes \Big(F_{\nabla^E} - \frac\iu\hbar\, \rho \cdot \One\Big) \bigg)
	\\[4pt]
	& \qquad \qquad \qquad = \bfdd_\rho \big( (E_0, \nabla^{E_0}) \otimes (E, \nabla^E) \big)\ .
\end{split}
\end{equation}
Here the sum is taken in the tangent space to $\HVBdl(M)$ at $(E, \nabla^E)$ -- it is different from the direct sum.
The higher covariant derivative is also compatible with the direct sum of hermitean vector bundles with connections.
Thus, in the sense of higher scalars, sections and functions it
satisfies all properties that one would expect of a derivative.

Covariantly constant higher functions $(E_0, \nabla^{E_0}) \in \Gamma(M, \Ia_0)$ are exactly the flat hermitean vector bundles with connection.
Therefore, constant higher functions are very different in general from higher scalars $V \in \Hilb$, because the collection of flat hermitean vector bundles on a manifold $M$ depends on the fundamental group $\pi_1(M)$ and thus depends strongly on the topology of $M$.
This is, however, just the next higher analogue of how the collection
of locally constant functions (i.e. those with vanishing de~Rham differential) are topological -- it is isomorphic to $H^0(M,\C)$, and thus detects $\pi_0(M)$.
On the contractible base space $M=\R^d$ higher constant functions and higher scalars agree up to equivalence of categories.

Turning back to the original motivation of obtaining momentum operators on the 2-Hilbert space $(\Gamma(M, \Ia_\rho), \<-,-\>)$ of global sections, we see that the appearance of the tangent category obscures the idea of seeing the covariant derivative $\bfdd_\rho$ as an \emph{operator} -- its source and target categories do not agree, whence we cannot straightforwardly interpret $\bfdd_\rho$ as an observable. A solution to this problem might be to exploit the fact that $\HVBdl$ defines a stack on manifolds, hence naturally comes equipped with a smooth structure, and modifying the above construction to work in suitable parameterised families. This presumably leads to a much more complex structure than what we have described in the present paper, and we leave it for future investigations.

\subsection*{Acknowledgements}
We thank Marco Benini and Lennart Schmidt for helpful discussions.
This work was supported by the COST Action MP1405 ``Quantum Structure of Spacetime'', funded by the European Cooperation in Science and Technology (COST).
The work of S.B. was partially supported by the RTG~1670 ``Mathematics Inspired by String Theory and Quantum Field Theory''.
The work of L.M. was supported by the Doctoral Training Grant ST/N509099/1 from the UK Science and Technology Facilities Council (STFC).
The work of R.J.S. was supported in part by the STFC Consolidated Grant ST/P000363/1 ``Particle Theory at the Higgs Centre''.

\newcommand{\etalchar}[1]{$^{#1}$}

\end{document}